\documentclass{article}

    \usepackage[preprint]{neurips_2023}

\usepackage[utf8]{inputenc} % allow utf-8 input
\usepackage[T1]{fontenc}    % use 8-bit T1 fonts
\usepackage{hyperref}       % hyperlinks
\usepackage{url}            % simple URL typesetting
\usepackage{booktabs}       % professional-quality tables
\usepackage{amsfonts}       % blackboard math symbols
\usepackage{nicefrac}       % compact symbols for 1/2, etc.
\usepackage{microtype}      % microtypography
\usepackage{xcolor}         % colors
\usepackage{sidecap}
\usepackage{paralist}
\usepackage{amssymb}
\usepackage{amsmath}
\usepackage{amsthm}
\usepackage{float}
\usepackage{graphicx}
\usepackage{caption}
\usepackage{subcaption}
\usepackage{bm}
\usepackage{algorithm}
\usepackage{algorithmic}
\usepackage{natbib}
\usepackage{comment}

\theoremstyle{plain}
\newtheorem{theorem}{Theorem}
\newtheorem{remark}{Remark}

\newtheorem{lemma}[theorem]{Lemma}

\theoremstyle{definition}
\newtheorem{definition}{Definition}[section]

\newcommand{\norm}[1]{{\left\vert\kern-0.25ex\left\vert\kern-0.25ex\left\vert #1
    \right\vert\kern-0.25ex\right\vert\kern-0.25ex\right\vert}}

\newcommand{\cf}{{\it cf. }}

\newcommand{\mesobomo}{{MESOB-OMO}}
\newcommand{\mesob}{{MESOB}}

\title{MESOB: Balancing Equilibria \& Social Optimality}

\author{
Xin Guo\thanks{Amazon, emails: \texttt{\{xnguo,llh,sareh,rabihsal\}@amazon.com}}
\thanks{Also affiliated with University of California Berkeley, IEOR Dept.}
\And
Lihong Li\footnotemark[1]
\And
Sareh Nabi\footnotemark[1]\,\,\footnotemark[4]
\And
Rabih Salhab\footnotemark[1]
\And
Junzi Zhang\thanks{Citadel Securities, work done while at Amazon, email: \texttt{saslas.c.royale@gmail.com}}\,\,\thanks{These authors contributed equally to this work and are co-first authors.}
}

\begin{document}

\maketitle

\begin{abstract}
Motivated by bid recommendation in online ad auctions, this paper considers a general class of multi-level and multi-agent games, with two major characteristics: one is a large number of anonymous agents, and the other is the intricate interplay between competition and cooperation. To model such complex systems, we propose a novel and tractable bi-objective optimization formulation with mean-field approximation, called \mesob\ (\textbf{M}ean-field \textbf{E}quilibria \& \textbf{S}ocial \textbf{O}ptimality \textbf{B}alancing), as well as an associated occupation measure optimization (OMO) method called \mesobomo\ to solve it. \mesobomo\ enables obtaining approximately Pareto-efficient solutions in terms of the dual objectives of competition and cooperation in MESOB, and in particular allows for Nash equilibrium selection and social equalization in an asymptotic manner. We apply \mesobomo\ to bid recommendation in a simulated pay-per-click ad auction. Experiments demonstrate its efficacy in balancing the interests of different parties and in handling the competitive nature of bidders, as well as its advantages over baselines that only consider either the competitive or the cooperative aspects.
\end{abstract}

\section{Introduction}
\label{intro}

The interplay between  competition and cooperation can be observed in a wide range of human interactions across various fields, such as economics, healthcare, education, sports, and politics, among others. Online advertising, which promotes product or service discovery through various online channels such as search engines, websites, and social media, is a prime example of the interplay between competition and cooperation. The ad market consists of multiple parties (advertisers, shoppers, and publisher) with their own interests. Advertisers compete in auctions for opportunities to show their ads to users visiting a website or using app. The highest bidder wins the ad placement to reach their desired audience, and is charged a price determined by a pricing mechanism. The publisher enriches its content portfolio with ads, and the revenue (paid by advertisers) ensures further content optimization that benefit users in the long term.

A common service offered to advertisers by publishers such as Amazon or Google is bid recommendations. The recommendation is based on various factors such as historical performance. These recommendations help advertisers make informed decisions and increase their chances of winning auctions, and ultimately improve user experience and publisher revenue. One key question for the ad recommendation service is how to make bid suggestions which are in the best interests of all parties involved, thus foster long-term success and satisfaction.

Bid recommendation can be viewed as a multi-level, multi-agent game, with two major characteristics: one is a large number of anonymous agents (a.k.a., advertisers) and the other the need to balance interests between the publisher and other parties involved including shoppers and advertisers. Recent development in mean-field theory inspires mean-field approaches to modeling the large number of agents, in order to avoid analyzing the otherwise computationally intractable multi-agent game. However, focusing the analysis mainly on the competition among advertisers is insufficient for balancing interests among multiple parties as illustrated in \S\ref{sec:bidding} and Appendix \ref{sec:empirical:baseline}. One approach is for the bid recommendation service to take the role of a \emph{social planner} with a dual goal of optimizing the social benefits of the shoppers, advertisers and the ad publisher, while maximizing the advertisers' individual interests given their inherent competitions.

\paragraph{Contributions.} Our contributions are three-fold. First, motivated by bid recommendation, we introduce a novel and tractable bi-objective optimization formulation (\mesob) (\textbf{M}ean-field \textbf{E}quilibria \& \textbf{S}ocial \textbf{O}ptimality  \textbf{B}alancing) for large-population systems where both competitive and cooperative interactions coexist. \mesob\ is derived from three key components: mean-field approximation for games with a large number of agents, the notion of exploitability for the analysis of Nash equilibria, and the notion of $(\epsilon_1, \epsilon_2)$-Pareto efficiency for the interplay between competition and cooperation.  
Second, we adopt the idea of occupation measure optimization (OMO) of \cite{guo2022mf} to translate the set of Nash equilibrium constraints into an equivalent and {\it finite} set of simpler constraints, and obtain  subsequently a constrained optimization problem called MESOB-OMO. This new formulation enables obtaining approximately Pareto-efficient solutions (Theorem~\ref{mesobomo-vs-pareto}), and in particular allows for  Nash equilibrium selection and social equalization in an asymptotic manner (Theorem~\ref{mesob-omo-asymptotic}).  
Finally, we apply MESOB-OMO to bid recommendation in simulated online pay-per-click ad auctions. Experiments show MESOB-OMO can balance the interests of different parties, while considering the competitive nature of bidders. Our results highlight the benefit of MESOB-OMO over existing approaches that solely focus on either the competitive or the cooperative aspects of the problem.

\paragraph{Related work.} 
The idea of mean-field approximation for games with a large number of agents has been widely adopted to approximate the otherwise generally intractable $N$-player games \citep{yang2020overview} since the pioneering work of \cite{huang2006large} and \cite{Lasry07_MFG}. Since then, the mean-field theory has evolved significantly. However, existing works focus on either the competitive (a.k.a. mean-field) games \citep{carmona2018probabilistic}
or the cooperative setting (a.k.a., mean-field control) 
\citep{pham2016discrete,gu2021mean}. 

The main research efforts on multi-level games are the principal-agent game (a.k.a., contract theory and mechanism design) \citep{sannikov2008continuous,elie2019tale} and Stackelberg game \citep{conitzer2006computing,fiez2019convergence,guo2022optimization}. Again, the primary optimization focus is the principal/leader thus not  bi-objective. 

The main concepts adopted to analyze the interplay between competition and cooperation for both $N$-player and mean-field games are the Price of Anarchy and Price of Stability \citep{roughgarden2007introduction,xu2019stochastic,bayraktar2021terminal,delarue2020selection,conitzer2022multiplicative}. However, these notions characterize mainly the gap between the social optimality and Nash equilibria, and not sufficient to analyze the balance between the social welfare and equilibria. 
 
Last but not least, there is a substantial literature studying bidding games in ad auctions. Again they mainly focus on bidder competitions, including  pacing equilibrium \citep{conitzer2022multiplicative} and system (Nash) equilibrium \citep{balseiro2017budget,IyeJohSun2014Mean,guo2019learning}.

\section{Problem formulation and solutions}
\label{sec:model}

In this section, we formulate and analyze a general mathematical framework of multiple-level\footnote{For a more precise definition, refer to Equation \eqref{bilevel} in \S\ref{sec:formulation}.} and multi-agent game. The framework is characterized by a large number of anonymous agents, and the need to balance interests between the social planner and other parties involved. In our motivating example of bid recommendation in ad auctions, the bid recommendation service plays the role of the social planner, and the advertisers are the anonymous agents.

\subsection{Problem formulation}
\label{sec:formulation}

\paragraph{Agent dynamics with mean-field interaction.} Given the large number of anonymous agents, we assume weak interaction among agents to simplify their interactions. Here, weak interaction means that agents interact solely through the population mean-field, which is the empirical distribution of their states and actions, as found in the classical mean field game literature~\citep{Lasry07_MFG,lasry2006jeux, lasry2006jeux2,Huang06_particles,huang2006large,saldi2018markov,gu2021mean}. This assumption streamlines the analysis by focusing on aggregated behavior rather than individual interactions.

With this mean-field approximation, one can now focus on any ``representative'' agent, and model her dynamics by a tuple $(\mathcal{S},\mathcal{A},\{P_t\}_{t=0}^T,\mu_0)$, where $\mathcal{S}$ and $\mathcal{A}$ are finite state and action spaces, respectively with $S:=|\mathcal{S}|<\infty$ and $A:=|\mathcal{A}|<\infty$, $P_t$ is the state-transition probability at time $t$, $\mu_0$ is the initial population state distribution, and $T\in[0,\infty)$ is a fixed time horizon. Furthermore, denote by $\Delta(\mathcal{X})$ the set of probability distributions on a set $\mathcal{X}$.

Given the state $s_t\in\mathcal{S}$ of the representative agent and her action $a_t \in \mathcal{A}$ at time $t$, and the population state-action joint distribution $L_t\in\Delta(\mathcal{S}\times\mathcal{A})$, her state at time $t+1$ follows the distribution $s_{t+1}\sim P_t(\cdot|s_t,a_t,L_t)$. By symmetry of the agents, we consider agents taking actions following the same time-dependent Markov policy, $\pi \in \mathcal{M}:=\{\{\pi_{t}\}_{t=0}^T|\pi_t:\mathcal{S}\to \Delta(\mathcal{A})\}$, where $\pi_t(a|s)$ is the probability of taking action $a$ in state $s$ at time $t$. Accordingly, the population state-action joint distribution $L\in\mathcal{L}:=\{\{L_t\}_{t=0}^T|L_t\in\Delta(\mathcal{S}\times\mathcal{A})\}$ initializes from $L_0(s,a)=\mu_0(s)\pi_0(a|s)$, and is recursively updated by $L_{t+1}(s',a')=\mu_t(s')\pi_t(a'|s')$, with $\mu_t(s')=\sum_{s\in\mathcal{S},a\in\mathcal{A}}L_t(s,a)P_t(s'|s,a,L_t)$. 
For notation simplicity, define $\Gamma:\mathcal{M}\rightarrow\mathcal{L}$ for the above recursive mapping such that  $\Gamma(\pi)_0(s,a):=\mu_0(s)\pi_0(a|s)$, and for $t=0,\dots,T-1$,
\begin{equation}\label{eq:gamma}
\Gamma(\pi)_{t+1}(s,a):=\pi_{t+1}(a|s)\sum_{s'\in\mathcal{S},a'\in\mathcal{A}}\Gamma(\pi)_t(s',a')P_t(s|s',a',\Gamma(\pi)_t).
%,\quad\for t=0,\dots,T-1.
\end{equation}
Unlike the standard mean-field framework, the agents here are {\it anonymous} but not necessarily uniformly {\it homogeneous}. Indeed, such a heterogeneity can be modeled by incorporating agent types into the state $s_t\in\mathcal{S}$ so that $s_t=(c,\tilde{s}_t)$, where $c$ is the (static) type  while $\tilde{s}_t$ is the dynamic state component \citep[Remark 2]{cui2021learning}.
In bid recommendation, for example, $c$ might represent the advertiser's ad quality and product category, while $\tilde{s}_t$ might denote the advertiser's reduced budget following the acquisition of an ad slot and payment to the publisher after a click. Throughout the paper, we assume  $\mu_0>0$ component-wise to allow for this type of heterogeneity. 
 
\paragraph{Agent rewards and Nash equilibrium.} Given the mean-field dynamics, at each time step $t$, the representative agent collects a reward $r_t(s_t,a_t,L_t)$, which depends on her current state $s_t$ and her action $a_t$, as well as the population distribution $L_t$. Her expected cumulative reward from an initial state $s_0=s$, policy $\pi\in\mathcal{M}$, and a population distribution flow $L\in\mathcal{L}$ is defined as 
$
J(s,\pi,L)=\mathbb{E}\left[\sum\nolimits_{t=0}^Tr_t(s_t,a_t,L_t)\right],
$
where $s_0=s$, $s_{t+1}\sim P_t(\cdot|s_t,a_t,L_t)$ for $t=0,\dots,T-1$ and $a_t\sim\pi_t(\cdot|s_t)$ for $t=0,\dots,T$. 

Here, we adopt the Nash equilibrium (NE) to analyze this mean-field game for the agents. A policy $\pi\in\mathcal{M}$ is called an NE for a mean-field game, if and only if there exists $L\in\mathcal{L}$, such that the best response condition and the consistency condition hold. That is,  
\[
J(s,\pi',L)\leq J(s,\pi,L),\quad\forall \pi'\in\mathcal{M},\,s\in\mathcal{S}.
\]
{\it and} $L=\Gamma(\pi)$, where the latter indicates that the population distribution flow is consistent with $\pi$ \citep{huang2006large,guo2019learning}.

\paragraph{Agent contribution to social optimality.} Given the impact of agents' dynamics over multiple stakeholders in the game, the social planner will incorporate agents' individual contributions to various social metrics into the overall objective, from which the social welfare is maximized. More precisely, consider $K>0$ relevant social metrics $V^{(k)}:\mathcal{L}\rightarrow\mathbb{R}$, where 
each social metric $V^{(k)}(L)$ for  $k=1,\dots,K$ at each time step $t=0,\dots,T$ is the aggregation of the agents' individual contributions $r_t^{(k)}(s_t,a_t,L_t)$, such that  
\begin{equation}\label{Vk_expansion}
 V^{(k)}(L):=\sum_{t=0}^T\sum_{s\in\mathcal{S},a\in\mathcal{A}}L_t(s,a)\ r_t^{(k)}(s,a,L_t).   
\end{equation}
The social welfare is then defined as 
\begin{equation}\label{V_formula}
V(L):=F(V^{(1)}(L),\dots,V^{(K)}(L)),
\end{equation}
where $F:\mathbb{R}^K\rightarrow\mathbb{R}$ is a function that connects the social metrics into a general form of utility functions. Since $L=\Gamma(\pi)$, 
this social welfare reduces to maximizing $V(\Gamma(\pi))$ over $\pi\in\mathcal{M}$. 

In bid recommendation, the social metrics may be the Return on Ad Spend (RoAS) for the advertiser, ad revenues for the publisher, and shopper clicks (as a proxy for shopper satisfaction). When $K=1$ and $F$ is the identity mapping, maximizing the overall social welfare is the classic mean-field control. This generalization of MFC enables general-utility use cases, such as RoAS in the bid recommendation application (refer to \S\ref{bid_recommend_example} for more details). 

\paragraph{Bi-level optimization problem for the social planner.} The social planner aims to maximize the welfare function $V(\Gamma(\pi))$ over the set of policies that are the NE of the agents' mean-field game. More precisely, it solves the following bi-level optimization problem: %That is, 
\begin{equation}\label{bilevel}
\begin{array}{ll}
\text{maximize}_{\pi\in\mathcal{M}} & V(\Gamma(\pi))\\
\text{subject to} & J(s,\pi',L)\leq J(s,\pi,L), \quad L=\Gamma(\pi), \quad \forall \pi'\in\mathcal{M},s\in\mathcal{S}.
 %\text{$\pi$ is an NE.}%
\end{array}
\end{equation}

However, this optimization problem is intractable, given the infinite number of constraints to characterize NE. More importantly, it suffers from the price of stability (PoS) \citep{roughgarden2007introduction}, namely, the intrinsic gap between the best achievable social welfare among the NEs and the overall best social welfare. We will develop an effective solution in the next section.

To address these issues, we will in the next section propose an alternative bi-objective optimization problem, by utilizing the notions of exploitability and $(\epsilon_1, \epsilon_2)$-Pareto efficiency. 
We then borrow the idea of  the occupation measure in \cite{guo2022mf} to translate the set of NE constraints into an equivalent and {\it finite} set of simpler constraints. The subsequent (single-objective) optimization problem is called MESOB-OMO, and we will show that it can be used to bridge the gap between social optimality and NEs with an arbitrarily desired trade-off (\textit{cf}. Theorem \ref{mesobomo-vs-pareto} and Theorem \ref{mesob-omo-asymptotic}).

\subsection{MESOB-OMO}\label{mesob-omo-main}

Our solution is built on the notions of exploitability and $(\epsilon_1, \epsilon_2)$-Pareto efficiency, to be defined shortly. We leverage the idea of the occupation measure in \cite{guo2022mf} to translate the set of NE constraints in \eqref{bilevel} into an equivalent and {\it finite} set of simpler constraints. The subsequent optimization problem, \mesobomo, is able to bridge the gap between social optimality and NEs with an arbitrarily desired trade-off (Theorems \ref{mesobomo-vs-pareto} and \ref{mesob-omo-asymptotic}).

\paragraph{\mesob: bi-objective problem of the social planner.} We first consider an alternative bi-objective optimization problem, by utilizing two notions. The first is the {\it exploitability} of a policy $\pi$:
\begin{equation}
\label{eq:exploitability}
    \text{Expl}(\pi):=\max_{\pi'\in\mathcal{M}}\sum_{s\in\mathcal{S}}\mu_0(s)\left[J(s,\pi',\Gamma(\pi))-J(s,\pi,\Gamma(\pi))\right],
 \end{equation}
which measures the distance of the policy $\pi$ from a Nash equilibrium~\citep{perrin2020fictitious}. By definition, $\text{Expl}(\pi) \geq 0$, and is equal to $0$ if and only if $\pi$ is an NE. Hence, finding NE is equivalent to minimizing exploitability. With this definition, the bi-level optimization \eqref{bilevel} is reformulated into one of maximizing $V(\Gamma(\pi))$ while minimizing $\text{Expl}(\pi)$, over all $\pi\in\mathcal{M}$. We refer to this problem as \textbf{MESOB} (\textbf{M}ean-field \textbf{E}quilibria \& \textbf{S}ocial \textbf{O}ptimality \textbf{B}alancing).

The second notion is $(\epsilon_1,\epsilon_2)$-Pareto efficiency, which characterizes the performance of a policy $\pi\in\mathcal{M}$ for the bi-objective optimization problem. 

\begin{definition}
\label{def:pareto}
A policy $\pi\in\mathcal{M}$ is called $(\epsilon_1,\epsilon_2)$-Pareto efficient for MESOB, if there is no policy $\pi'\in\mathcal{M}$ such that: $V(\Gamma(\pi'))\geq V(\Gamma(\pi))+\epsilon_1$, $\text{Expl}(\pi')\leq \text{Expl}(\pi)-\epsilon_2$, and at least one of the two inequalities is strict.
\end{definition}
Note that $(0,0)$-Pareto efficiency corresponds to the standard Pareto-efficiency. It also generalizes the notion of $\epsilon$-Pareto-efficiency~\citep{liu1996epsilon}.

It is tempting to find directly a policy $\pi$ that minimizes the scalarization of the following bi-objective optimization problem 
\begin{equation}\label{scalarization}
\text{minimize}_{\pi\in\mathcal{M}}\, -\lambda_1 V(\Gamma(\pi)) + \lambda_2 \, \text{Expl}(\pi),
\end{equation}
where suitable constants $\lambda_1,\lambda_2>0$. However, the exploitability term $\text{Expl}(\pi)$ is neither convex nor smooth, even when the rewards and the dynamics are smooth with respect to the population distributions $L_t$, rendering \eqref{scalarization} difficult to solve. Instead, we propose  an occupation-measure-based optimization (OMO) approach to compute the $(\epsilon_1,\epsilon_2)$-Pareto efficient solution for \mesob, which we call \mesobomo, that addresses the smoothness issue.

\paragraph{\mesobomo.} 
Our approach is inspired by the optimization reformulation of mean-field games \citep{guo2022mf}, in which an occupation measure variable $d\in\mathbb{R}^{SA(T+1)}$ is adopted to  approximate the population distribution flow with $d_{t,s,a}\approx L_t(s,a)$. This occupation measure was first adopted for the linear program formulation of MDPs by \cite{manne1960linear}. 

Besides the occupation measure, there are two additional critical components $g^{\text{CS}}(d)$ and $h^{\text{BR}}(y,z,d)$, to be defined later. Intuitively, $g^{\text{CS}}(d)$ encodes the consistency condition of a mean-field NE, and $h^{\text{BR}}(y,z,d)$ encodes the best-response condition of NE. 

With these key components, MESOB can be reformulated as the following optimization problem that optimizes over an occupation measure variable $d\in\mathbb{R}^{SA(T+1)}$, and auxiliary variables $y\in\mathbb{R}^{S(T+1)}$ and $z\in\mathbb{R}^{SA(T+1)}$: 
\begin{equation}\label{mesobomo}
    \begin{array}{ll}
        \text{minimize}_{y,z,d} &
        f^{\texttt{MESOB-OMO}}(y,z,d;\lambda_1,\lambda_2,\rho_1,\rho_2)\\
        %&\,\,\,:= ,\\
        \text{subject to} & \{\vec{d}_{t}\}_{t=0}^T\subseteq\Delta(\mathcal{S}\times\mathcal{A}),\quad z\in Z,\\
        &\|y\|_2\leq \frac{S(T+1)(T+2)r_{\max}}{2}.
        \tag{\mesobomo}
        %L\geq 0,\quad z\geq 0, \\
        %& {\bf 1}^\top z\leq SA(T^2+T+2)r_{\max}, \\  & %\|y\|_2\leq S(T+1)(T+2)r_{\max}/2.
    \end{array}
\end{equation}

The detailed step-by-step derivation of {MESOB-OMO} will be given in the appendix.
Here, $f^{\texttt{MESOB-OMO}}(y,z,d;\lambda_1,\lambda_2,\rho_1,\rho_2):=-\lambda_1 V(d)+\lambda_2z^\top d+\rho_1g^{\text{CS}}(d)+\rho_2 h^{\text{BR}}(y,z,d)$.
$\lambda_1,\,\lambda_2>0$ and $\rho_1,\,\rho_2>0$ are two pairs of user-defined hyper-parameters. The former serves as trade-off parameters for the dual objectives, while the latter as penalty parameters to enforce the consistency and best-response conditions, respectively, thus approximating exploitability. In addition, $\vec{d}_t\in\mathbb{R}^{SA}$ is the $t$-th slicing of the tensor in its time dimension (and accordingly, $\vec{d}_t\approx L_t$).

Specifically, $Z:=\left\{z\,|\,z\geq0,\, {\bf 1}^\top z\leq SA(T^2+T+2)r_{\max}\right\}$, and 
$y\in\mathbb{R}^{S(T+1)}$ and $z\in\mathbb{R}^{SA(T+1)}$ are auxilliary variables. Moreover,
{
\[
g^{\text{CS}}(d) := \sum\limits_{s\in\mathcal{S}}\left(\sum\limits_{a\in\mathcal{A}}d_{0,s,a}-\mu_0(s)\right)^2 \!\!
+\sum\limits_{s'\in\mathcal{S}}\sum\limits_{t=0}^{T-1}\left(\sum\limits_{a\in\mathcal{A}}d_{t+1,s',a}-\!\!\!\!\sum\limits_{s\in\mathcal{S},a\in\mathcal{A}} \!\!\!\! d_{t,s,a}P_t(s'|s,a,\vec{d}_{t})\right)^2\!\!,
\]
}
and
{
\[
\begin{split}
\lefteqn{h^{\texttt{BR}}(y,z,d) := \sum\limits_{s\in\mathcal{S},a\in\mathcal{A}}\left(y_{T-1,s}-r_T(s,a,\vec{d}_{T})-z_{T,s,a}\right)^2} \\
&+\sum\limits_{s\in\mathcal{S},a\in\mathcal{A}}\sum\limits_{t=0}^{T-2}\left(y_{t,s}-r_{t+1}(s,a,\vec{d}_{t+1})-\sum_{s'\in\mathcal{S}}P_{t+1}(s'|s,a,\vec{d}_{t+1})y_{t+1,s'}-z_{t+1,s,a}\right)^2\\
&+\sum\limits_{s\in\mathcal{S},a\in\mathcal{A}}\left(y_{T,s}+r_0(s,a,\vec{d}_{0})+\sum\limits_{s'\in\mathcal{S}}P_0(s'|s,a,\vec{d}_{0})y_{0,s'}+z_{0,s,a}\right)^2.
\end{split}
\]
}

Note that the constraints on $y,z,d$ are independent and simple convex constraints with closed-form projection formulas. These constraints can be easily removed by softmax and trigonometric parameterizations without compromise on smoothness. 

Once an approximate solution to \eqref{mesobomo} is obtained, a policy $\pi\in\Pi(d)$ can be retrieved by defining a set-valued mapping $\Pi$ to retrieve policies from an occupation measure $d$ by normalization: given $d\in\mathbb{R}^{SA(T+1)}$ with $d\ge 0$, $\pi_t(a|s)=\frac{d_{t,s,a}}{\sum_{a'\in\mathcal{A}}d_{t,s,a'}}$ when $\sum_{a'\in\mathcal{A}}d_{t,s,a'}>0$, and $\pi_t(\cdot|s)$ is an arbitrary probability vector in $\Delta(\mathcal{A})$ otherwise.

\begin{remark}
Evidently, smoothness of $r_t$, $r_t^{(k)}$ and $P_t$ ($t=0,\dots,T$, $k=1,\dots,K$) in $L_t$ are inherited by the \texttt{MESOB-OMO} objective $f^{\text{MESOB-OMO}}(y,z,d;\lambda_1,\lambda_2,\rho_1,\rho_2)$. Thus MESOB-OMO can be solved by a large number of existing optimization algorithms, such as projected gradient descent, provided that the rewards and dynamics are smooth in $L_t$.  
\end{remark}

Now, denote $f^\star(\lambda_1,\lambda_2,\rho_1,\rho_2)$ as the optimal objective value of \eqref{mesobomo}. It is connected with the original dual objective target in \eqref{def:pareto}, as follows.

\begin{theorem}\label{mesobomo-vs-pareto}
Suppose that the link function $F$ in \eqref{V_formula} is Lipschitz continuous, and that for any $t=0,\dots,T$ and $k=1,\dots,K$, $P_t(\cdot|\cdot,\cdot,L_t)$, $r_t(\cdot,\cdot,L_t)$, and $r_t^{(k)}(\cdot,\cdot,L_t)$ are Lipschitz continuous in $L_t$. Then, there exists a constant $C>0$ that depends only on the problem data, such that for any target tolerance $\epsilon>0$, if one sets $\rho:=\min\{\rho_1,\rho_2\}>C\frac{\lambda_2^2\max\{\lambda_1,\lambda_2,1\}}{\epsilon^2}$ and compute a feasible solution $(y,z,d$) for \eqref{mesobomo} with $f^{\texttt{MESOB-OMO}}(y,z,d;\lambda_1,\lambda_2,\rho_1,\rho_2)-f^\star(\lambda_1,\lambda_2,\rho_1,\rho_2)\le\epsilon$, then any policy $\pi\in\Pi(d)$ is $(\epsilon/\lambda_1, \epsilon/\lambda_2)$-Pareto efficient for \mesob. 
\end{theorem}

This theorem guarantees the flexibility of balancing between the social optimality and the Nash equilibria, by choosing $\lambda_1$ and $\lambda_2$ appropriately. Moreover, in order to obtain a target $(\epsilon_1,\epsilon_2)$-Pareto efficient solution to \mesob, one can first set $\lambda_1=\epsilon/\epsilon_1$, $\lambda_2=\epsilon/\epsilon_2$ and $\rho_1,\rho_2=\Omega\left(\max\left\{\frac{\epsilon}{\epsilon_2^2\min\{\epsilon_1,\epsilon_2\}},\frac{1}{\epsilon_2^2}\right\}\right)$, then solve \mesobomo\ to obtain an $\epsilon$-sub-optimality. 

One can further characterize the solution to \mesobomo\ with respect to the trade-off ratio $\lambda_1/\lambda_2$, which, as a by-product, also connects \eqref{mesobomo} with the original optimization problem \eqref{bilevel}. 

\begin{theorem}\label{mesob-omo-asymptotic}
Given the same assumptions in Theorem \ref{mesobomo-vs-pareto} on $F,\,P_t,\,r_t$ and $r_t^{(k)}$, consider a sequence of %uniformly bounded 
constants $\{\lambda_1^l,\lambda_2^l,\epsilon^l\}_{l=0}^\infty$. For any  $l \in \{0,1,\ldots\}$, let $\rho^l:=2\max\{(\lambda_1^l)^2,\lambda_1^l\lambda_2^l,\lambda_1^l,\lambda_2^l,1\}/\epsilon^l$ and $(y^l,z^l,d^l)$ be a feasible solution to \eqref{mesobomo} with 
$f^{\texttt{MESOB-OMO}}(y^l,z^l,d^l;\lambda_1^l,\lambda_2^l,\rho^l,\rho^l)-f^\star(\lambda_1^l,\lambda_2^l,\rho^l,\rho^l)\leq \epsilon^l$. Then the following results hold. 
\begin{itemize}
    \item If  $\lim\limits_{l\rightarrow\infty}\frac{\lambda_1^l}{\lambda_2^l}=0$, $\lim\limits_{l\rightarrow\infty}\epsilon^l=0$, and $\inf\limits_{l\geq 0}\lambda_1^l>0$ then for any limit point $\bar{d}$ of $d^l$ and any $\bar{\pi}\in\Pi(\bar{d})$, $\bar{\pi}$ maximizes $V(\Gamma(\pi))$ over all NE policies $\pi$ with $\text{Expl}(\pi)=0$.  
    \item If  $\lim\limits_{l\rightarrow\infty}\frac{\lambda_1^l}{\lambda_2^l}=\infty$, $\lim\limits_{l\rightarrow\infty}\epsilon^l=0$, and $\inf\limits_{l\geq 0}\lambda_2^l>0$, then for any limit point $\bar{d}$ of $d^l$ and any $\bar{\pi}\in\Pi(\bar{d})$, $\bar{\pi}$ minimizes $\text{Expl}(\pi)$ over all socially optimal policies that maximizes $V(\Gamma(\pi))$.   
\end{itemize} 
\end{theorem}

That is, \mesobomo\ is asymptotically an equilibrium selection problem, finding an NE with the largest social objective as $\lambda_1/\lambda_2\rightarrow0$, and asymptotically a social equalizing problem, selecting a social optimal policy with the minimum exploitability as $\lambda_1/\lambda_2\rightarrow\infty$. 

\raggedbottom
\section{Case study: bidding in ad auctions}\label{bid_recommend_example}

In this section, we apply \mesob\ to bid recommendation in a simulated online pay-per-click ad auctions. Our focus is on how \mesobomo\ balances stakeholders' welfare, while managing the competition among advertisers.

\subsection{Bid recommendations for advertisers} \label{sec:bidding}

In online advertising, advertisers compete for supplies (e.g., slots on a webpage), shoppers seek high-quality content or services, and ad publisher makes a profit from advertising revenue. The ad recommendation service, as a social planner, provides bid suggestions by considering competition among advertisers, and the broader interests of shoppers, advertisers, and the publisher. Given our focus on developing a general framework for multi-level multi-agent games, our simulation assumes for simplicity that advertisers adopt the publisher's bid recommendations \footnote{In reality, advertisers may not always follow the recommendations. We may incorporate a propensity model in our optimization framework, which predicts how likely an advertiser follows the publisher's recommendation.}.

\textbf{Bidders' mean-field game.} We model the competition among bidders as a non-cooperative repeated mean-field game. At each time step $t\in\{0,1,\ldots,T\}$, a shopper makes a request (e.g., visiting a page) leading to a supply, and a set of $n$ bidders' states are sampled i.i.d. from a distribution $\mu$ to participate in an auction. Following our mean-field approximation described in \S\ref{sec:formulation}, a representative bidder at state $s_t$ places a bid $a_t \sim \pi(.|s_t)$ and receives reward $r(s_t,a_t,L_t)$, where $L_t$ is the population joint state-action distribution at $t$. Each bidder's state at time $t$ corresponds to its ad's CTR at that particular time. In this context, we assume there is no state transition.

Bidders compete through a second-price auction \citep{edelman2007internet}, where the winning bidder has the highest score defined as $\text{bid} \times \text{CTR}$. When clicked, the winner pays the second-highest score divided by the winner's CTR. In case of ties, the winner is chosen randomly and pays his placed bid. Here, $n$ is referred to as the auction density and its higher values are indicative of intensified competition. Detailed formulation of $r(s_t,a_t,L_t)$ is given in Appendix \ref{sec:reward&socialcomp}. 

\textbf{Bi-objective optimization of the bid recommendation service.} The bid recommendation service must balance the interests of all stakeholders when suggesting bids. To this end, we propose a welfare function that represents the objectives of advertisers, shoppers, and the ad publisher. We consider advertisers' primary goal to be optimizing the Return on Ad Spend (RoAS), denoted as the revenue per advertising dollar spent \footnote{As outlined in \S\ref{sec:bidding}, the rewards of bidders, $r(s_t,a_t,L_t)$, are derived from their participation in auctions. Different bidders may prioritize varying objectives for their advertising campaigns such as maximizing expected cumulative rewards or maximizing their campaigns' Return on Advertising Spend (RoAS). In this study, we employ $r(s_t,a_t,L_t)$ in the context of bidders' mean-field game for estimating the Nash equilibrium, and use RoAS to signify the advertiser's interest within the overall welfare component.}. On the other hand, the publisher's objective is revenue maximization, measured as the sum of Cost-per-Click (CPC) -- the payment received from the winning bidder when an ad is clicked. Shoppers are primarily interested in ads relevant to their search intent. In this study, we use CTR as a surrogate for shopper satisfaction.

We define $\textsc{CTR}(s,a,L)$, $\textsc{SALE}(s,a,L)$, and $\textsc{CPC}(s,a,L)$ as the expected CTR, Sale, and CPC of the representative bidder at state $s$, taking action $a$, considering mean field $L=\{L_t\}_{t\in \mathcal{T}}$. Explicit formulas are provided in Appendix \ref{sec:reward&socialcomp}. Omitting $L_t$ arguments for brevity, the market average of these entities are as follows:
{\footnotesize{
\[
V^{(1)}(L) \!=\! \sum_{t,s,a} L_t \cdot \textsc{CTR}(s,a,L_t),\ V^{(2)}(L) \!=\! \sum_{t,s,a} L_t \cdot \textsc{SALE}(s,a,L_t),\ V^{(3)}(L) \!=\! \sum_{t,s,a} L_t \cdot \textsc{CPC}(s,a,L_t).
\]
}}
The welfare function is defined as $V(L) = F(V^{(1)}(L),V^{(2)}(L),V^{(3)}(L)),$
where 
\begin{equation}\label{eq:welfare}
    \small
    F(V^{(1)}, \! V^{(2)}, \! V^{(3)})=c_1 \ln\!\left(V^{(1)}+\epsilon_1\right) \!
    +c_2 \ln\!\left(\frac{V^{(2)}}{V^{(1)}\times V^{(3)}+\epsilon_0}+\epsilon_2\right) \! +c_3 \ln\!\left(V^{(1)}\times V^{(3)}+\epsilon_3\right).
\end{equation}
The first term in \eqref{eq:welfare} denotes the interests of shoppers, while the second and third terms respectively refer to the interests of advertisers and the publisher. The positive scalar values, $\epsilon_i$'s, are used to prevent division by zero or taking logarithm of zero. The coefficients $c_i$ are positive.

The bid recommendation service operates as a social planner with the goal of solving a bi-objective optimization problem. It seeks a bidding policy $\pi \in \mathcal{M}$, which maximizes the welfare function $V(\Gamma(\pi))$, whilst ensuring competitive bids approximate a Nash equilibrium \footnote{As described in \S\ref{sec:formulation}, $\Gamma(\pi)$ denotes the distribution flow under policy $\pi$, with $\Gamma(\pi)(s,a) = \mu(s) \pi(a|s)$, given our assumption of no state transition in our bid recommendation use case (refer to \eqref{eq:gamma} for more details).}. We have employed our \mesob\ framework to model this application, in which the bid recommendation service explores for an $(\epsilon_1, \epsilon_2)$-Pareto efficient bidding policy, as defined in \eqref{def:pareto}. The forthcoming section will present the application of our MESOB-OMO formulation \eqref{mesobomo} in estimating this bidding policy.

\paragraph{An equilibrium-agnostic heuristic.}
For comparison with MESOB, we first consider a heuristic that determines bid recommendations based on a window of past auctions, without considering equilibria.
This heuristic observes the population bidding and CTR distributions at each timestep, simulates a certain number of second-price auctions to approximate the distribution of winning bids, and then proposes the 25th and 75th percentiles of this distribution as the lower and upper bidding ranges, respectively.\footnote{Like MESOB, this heuristic assumes all bidders comply with the bid recommendations.} A comprehensive description of this heuristic is provided in Appendix \S\ref{sec:empirical:baseline}. 

We then further apply existing algorithms that estimate the mean-field equilibrium within a competitive mean-field game environment, without considering social welfare. They include Mean-Field Occupation Measure Optimization (MF-OMO) \citep{guo2022mf}, Online Mirror Descent (OMD) \citep{perolat21}, and Fictitious Play (FP) \citep{perrin2020fictitious}. We then evaluate and compare the interests of all stakeholders under the equilibrium bidding policy versus the agnostic bidding heuristic. Our experiments show that the mean-field equilibrium bidding policy increases competition at bidding, potentially benefiting the ad publisher and shoppers at the expense of advertisers. 
For example, MF-OMO results in a lower exploitability of $0.0046$, compared with the heuristic policy's $0.2643$, a $98\%$ improvement.
% \footnote{These values originate from these parameter settings: $n=5$, $v=2$, $\eta=0.7$, $\kappa=10$, $\text{CTR}\in[0.2,0.6]$, and $\text{bids}\in[0,5]$ discretized into 3 CTRs and 5 bids (equally spaced), and $T=1000$ averaged over 1000 independent runs.}. 
Details are in Appendix \ref{sec:empirical:baseline}.
This significant reduction underscores the efficiency of MF-OMO in generating a bidding policy approximating an equilibrium, as expected.

Next, we solve \eqref{mesobomo} for our bidding application, and derive \mesob\ bidding policy that effectively navigates the welfare trade-off among all stakeholders while being in close proximity to the mean-field Nash equilibrium. 

\subsection{Experiment results with \mesob}
\label{exp:mesob}

We will now present experimental results for \mesob. We will implement \mesobomo\ using MFGlib \citep{guo2023mfglib}. Our discussion commences with an overview of the Pareto efficiency curve, corresponding to various $\lambda_1$ and $\lambda_2$ configurations of the \mesob\ parameters. This curve demonstrates the balance between the two main objectives in our bidding use case: overall social welfare and exploitability. Subsequently, we draw a Pareto curve for each constituent of our social welfare: shoppers, bidders, and the ad publisher. We will conclude this section by presenting \mesob\ bidding policy for a particular parameter setup.

In our experiments, we utilize a tabular setting where a bidder's state $s\in \mathcal{S}$ and action $a \in \mathcal{A}$ correspond to her CTR and placed bid, respectively. The state space $\mathcal{S}$ is derived by discretizing $[0.2,0.6]$ with equal spacing, and action space $\mathcal{A}$ by discretizing $[0,a_{max}]$ with equal spacing, where $a_{max}$ denotes the maximum bid allowed, set at $a_{max}=\$5$ in all experiments. We initialize parameters as follows: $n=5$, $\text{CTR}\in \{0.2,0.4,0.6\}$, $\text{bid}\in\{0, 1.25, 2.5, 3.75, 5\}$ (in dollars), $n_s=3$, $n_a=5$, and utility $v=2$\footnote{As detailed in \S\ref{sec:bidding}, bidders are anonymous with a known fixed utility reflecting their willingness to pay for an ad slot. The representative bidder's utility, denoted as $v$, can be approximated by $\text{price} \times \text{CVR}$, where $\text{price}$ signifies the average price of an advertised product and conversion rate $\text{CVR}$ represents the sale likelihood. We assume that $\text{price} \times \text{CVR}$ is known or can be estimated from offline data if necessary.}. The $\epsilon_i$ and $c_i$ values in \eqref{eq:welfare} are assigned as $\epsilon_i=10^{-5}, c_i=1/3$ for all $i$. We fine-tune the penalization parameters to $\rho_1=1$ and $\rho_2=0.1$ in \eqref{mesobomo}. Lastly, the lack of randomness in our set up negates the need for averaging across multiple runs. 
%{\color{red}{For the stopping criteria, we yield a policy $\pi$ that minimizes $-\lambda_1V(\Gamma(\pi))+\lambda_2\text{Expl}(\Gamma(\pi))$, thus optimizing the balance between social welfare and exploitability for any $\lambda_1$ and $\lambda_2$.}}
% {\color{red}{Regarding the stopping criteria, for any $\lambda_1$ and $\lambda_2$ parameter configurations, we generate a policy $\pi$ that minimizes the expression $-\lambda_1V(\Gamma(\pi))+\lambda_2\text{Expl}(\Gamma(\pi))$ in order to achieve an optimal balance between social welfare and exploitability.}} 
% 
\paragraph{Pareto efficiency curves.}
Figures~\ref{fig:pareto_overall} and \ref{fig:pareto_component} present the Pareto efficiency curve in relation to our two objectives: minimizing exploitability and maximizing social welfare. Figure\ref{fig:pareto_overall} illustrates the overall welfare, while Figure~\ref{fig:pareto_component} represents each component of the welfare, encompassing the welfare of shoppers, bidders, and the ad publisher. Throughout, we have varied the MESOB parameters, $\lambda_1$ and $\lambda_2$ in \eqref{mesobomo}. We observe that as   $\lambda_1/\lambda_2$ increases (i.e., prioritizing social welfare), the ad publisher welfare increases, bidders' welfare declines, and shoppers' welfare experiences a slight incline, collectively leading to an increase in overall welfare. The result shows that MESOB adeptly navigates the balance between collaboration and competition inherent in our bidding use case.

In a pure competition scenario, where $\lambda_1=0$ and $\lambda_2=1$ (marked by the red data point in both plots), MESOB generates an optimal bidding policy with the lowest exploitability and social welfare. This scenario is computed using MF-OMO algorithm to solve the corresponding mean field game\footnote{We employed other established algorithms including Online Mirror Descent (OMD) \citep{perolat21} and Fictitious Play (FP) \citep{perrin2020fictitious} to solve this mean-field game. Exploitability and social welfare values were consistent with those achieved through MF-OMO.}. Conversely, in a pure collaboration scenario, akin to general utility mean field control (MFC), where $\lambda_1=1$ and $\lambda_2=0$ (indicated by the black data point in both plots), MESOB mirrors the solution from the general utility MFC. This results in the highest social welfare at the expense of high exploitability. Notably, we have employed a gradient descent approach to solve the general utility MFC.

Within the MESOB framework and parameters $\lambda_1,\lambda_2\in(0,1)$, the system adeptly tailors an optimal policy balancing collaboration and competition. Each blue data point represents a \mesob\ solution, their distribution reflecting the variance with emphasis on collaboration versus competition as $\lambda_1,\lambda_2$ oscillate within the specified range. As observed from the plots, pure competition results in the lowest exploitability, while pure collaboration yields the highest social welfare, albeit with increased exploitability. \mesob\ adeptly balances these two objectives, allowing for a coexistence of collaboration and competition. For all these runs, we utilize gradient descent to solve \eqref{mesobomo}.
\begin{SCfigure}
    \centering   
    \caption{
    Pareto curve varying  $\lambda_1$ and $\lambda_2$ in \eqref{mesobomo}, trading off exploitability (x-axis) and social welfare (y-axis).
    The red point, solved by MF-OMO, corresponds to pure competition ($\lambda_1=0, \lambda_2=1$). The black point represents a scenario entirely focused on social welfare, ($\lambda_1=1, \lambda_2=0$). Blue points denote a range of configurations with $0<\lambda_1,\lambda_2<1$.}
    \label{fig:pareto_overall}
    \includegraphics[width=0.42\textwidth]
    {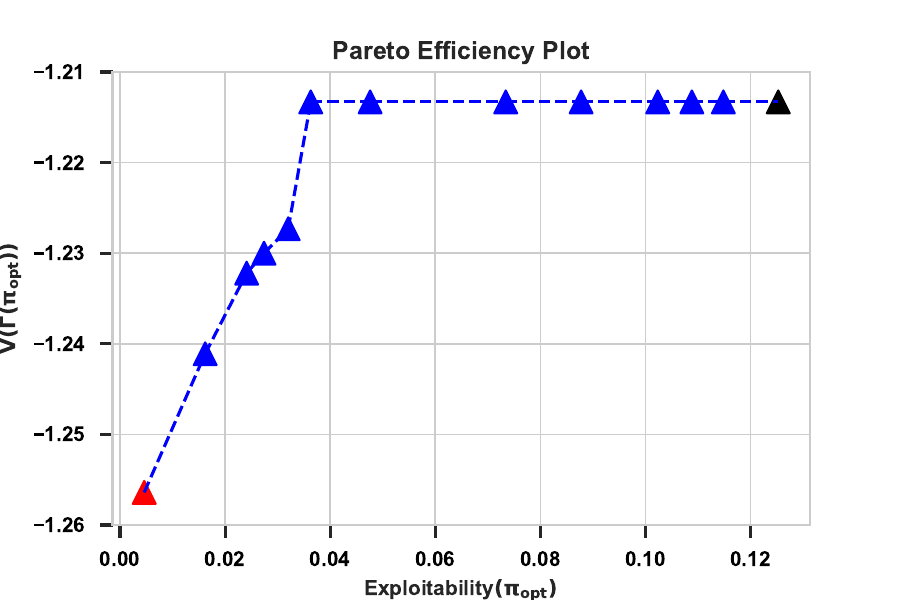}
\end{SCfigure}
\begin{figure}[htbp]
    \centering
    \includegraphics[width=0.85\textwidth]{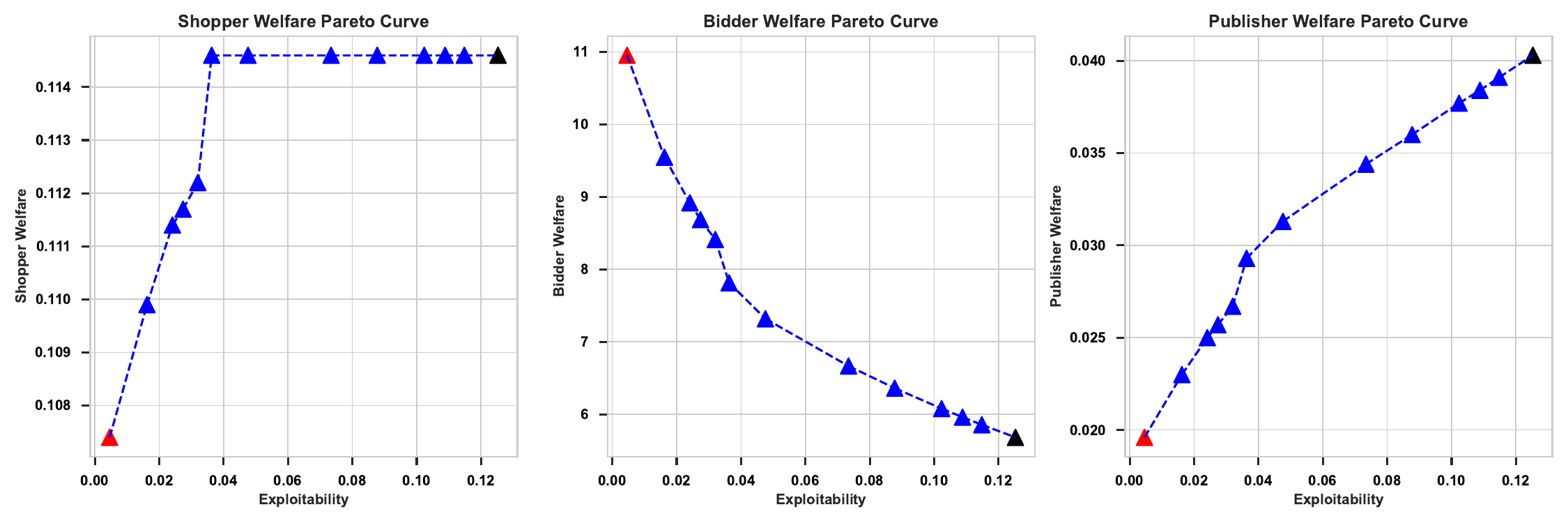}
    \caption{
    Pareto curve for welfare components (shoppers, bidders, ad publisher) under diverse MESOB parameter configurations. The NE solution ($\lambda_1=0, \lambda_2=1$) is red, while the MFC solution ($\lambda_1=1, \lambda_2=0$) is black. Blue points represent MESOB solutions for $\lambda_1,\lambda_2\in(0,1)$.}    \label{fig:pareto_component}
\end{figure}
\paragraph{\mesob\ execution.} We illustrate a particular execution of MESOB with $\lambda_1=\lambda_2=0.5$, represented by one of the blue triangles in Figure~\ref{fig:pareto_overall}.
Figure~\ref{fig:mesob_run} shows MESOB bidding policy at different states obtained by optimizing \eqref{mesobomo} for 1500 iterations.
\noindent\begin{minipage}{.5\textwidth}
    \centering
    \includegraphics[width=.9\linewidth]{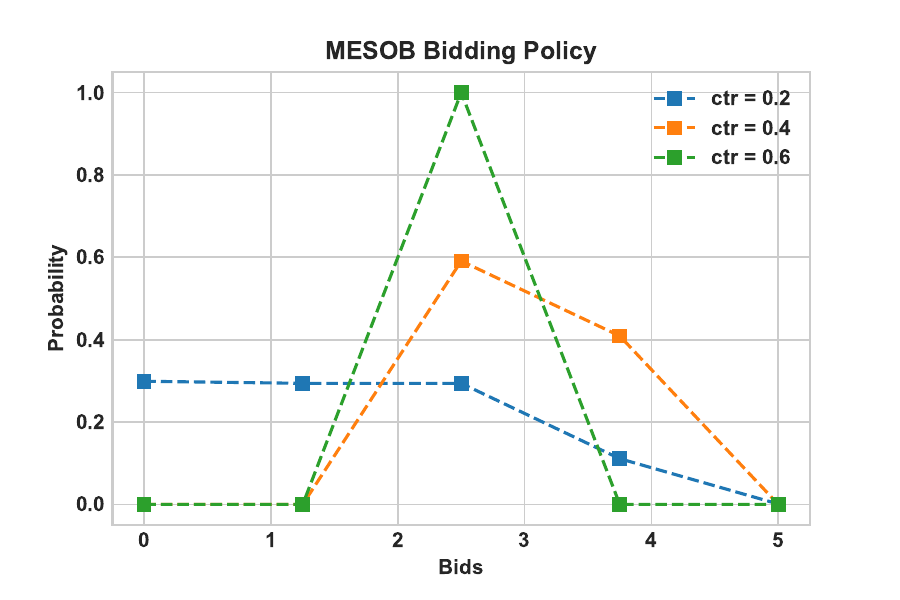}
\end{minipage}%
\begin{minipage}{.5\textwidth}
    \centering
    \captionof{figure}{\mesob\ bidding policy for $\lambda_1=\lambda_2=0.5$ (\eqref{mesobomo}). It outlines bid recommendations and their probabilities at different click-through rates (ctr). At $ctr=0.6$, a $\$2.5$ bid is suggested, while at $ctr=0.4$, bids of $\$2.5$ with probability $0.59$ and $\$3.75$ with probability $0.41$ are suggested. At $ctr=0.2$, bids of $\$0,\$1.25,\$2.5,\$3.75$ are proposed with respective probabilities of $0.3, 0.29, 0.29, 0.11$.}
    \label{fig:mesob_run}
\end{minipage}

\raggedbottom
\section{Conclusions}
\label{conclusion}

In applications such as bid recommendations for online ad auctions, it is crucial that recommended bids account for the competition among advertisers for online ad impressions, as well as the interests of the ad publisher, shoppers, and advertisers. A bid recommendation service seeks to strike an optimal balance among all stakeholders' interests within its multi-level multi-agent game system. Numerous applications exist where such a complex interplay of competition and collaboration coexists. Traditional social optimality or equilibrium concepts alone fall short in addressing this interplay. This paper proposes the novel \mesob\ (\textbf{M}ean-field \textbf{E}quilibria \& \textbf{S}ocial \textbf{O}ptimality \textbf{B}alancing) framework and the \mesobomo\ tractable formulation by adopting the idea of occupation measure optimization (OMO). These tools effectively navigate the intricate balance between competition and cooperation in multi-level multi-agent systems, employing mean-field techniques to address the large volume of interactions. We establish convergence guarantees for \mesobomo\ and derive its asymptotic relationship with the classical equilibrium selection and social equalization problems. The empirical efficacy of this approach is showcased through the estimation \mesob\ bidding policies that balance the interests of all stakeholders while addressing the competition among advertisers in a simulated pay-per-click ad auction environment.

\section{Acknowledgments}
\label{ack}

The authors appreciate Ben Allison, Nikhil Devanur, Barath Ezhilan, Aditya Maheshwari, Muthu Muthukrishnan, Daniel Oliveira, Sergio Rodriguez, Koushiki Sarkar, Amin Sayedi, Kiarash Shaloudegi, and Ziyang Tang for insightful feedback and support.

\bibliography{refs}
\bibliographystyle{abbrvnat}

\appendix 

\section*{Appendix}
\addcontentsline{toc}{section}{Appendix}

\section{Reward formulation in bidders' mean-field game} 
\label{sec:reward&socialcomp}

Here we provide closed-form expressions for the bidder's reward $r(s_t,a_t,L_t)$ and the social welfare components $V^{(i)}$ ($i=1,2,3$) as described in \S\ref{sec:bidding}. 

Let $\zeta$ represent the set of all possible scores, defined as $\zeta:=\{z \mid z=s\times a, \forall s\in\mathcal{S}, a\in\mathcal{A}\}$. Scores are ordered such that $z_0\leq z_1\leq \dots \leq z_{max}$. Furthermore, let $\lambda_t(z)$ denote the probability of observing a score $z$ as $\lambda_t(z)=\sum_{(s',a'):\ s'a'=z} L_t(s',a')$, and $\Lambda_t(z)$ represent the likelihood of observing a score less than $z$, defined as $\Lambda_t(z)= \sum_{j=0}^{l-1} \lambda_t(j)$. Here, the score $z=s\times a$ is assumed to be the $l$-th element in $\zeta$ (i.e., $z_l=z$). For simplicity, we discretize both \text{bid} and \text{CTR} values to ensure that any score $z$ corresponds to a unique $(s,a)$ pair such that $z=s\times a$, so $\lambda_t(z)=L_t(s,a)$.

In each auction, the representative bidder competes against $(n-1)$ other bidders. At any time $t$, given the mean-field flow $L_t$, the probability of winning the auction for the representative bidder at state $s$ and placing bid $a$ can be determined in two distinct scenarios. The first scenario occurs when there is no tie (a single winner). For simplicity, we assume that the representative bidder score $z=s\times a$ is the $l$-th element among the ordered scores, and the winning probability in this case can be obtained by summing over all cases where the second highest score ($z_j$) is any score $z_0\le z_j<z_l$. The probability in this scenario is given by $P_1$ in \eqref{eq:win_prob}. The second scenario arises when more than one bidder shares the highest score. To obtain the winning probability in this case, we need to sum over all possibilities of $1$ up to $(n-1)$ bidders sharing the same highest score as given in $P_2$.
\begin{equation}\label{eq:win_prob}
P_1 = \sum_{j=0}^{l-1} \big(\Lambda_t(z_{j+1})^{n-1}-\Lambda_t(z_j)^{n-1}\big), \quad
P_2 = \sum_{i=1}^{n-1} \frac{1}{i+1} \ \binom{n-1}{i}\ \Lambda_t(z)^{n-i-1}\ \lambda_t(z)^i.
\end{equation}
Next, we compute the expected values of CTR, CPC, and sale for the representative bidder using winning probabilities in \eqref{eq:win_prob}. If a click occurs (with probability $s$), the bidder pays $\frac{z_j}{s}$ for a solo win, or his bid $a$ if there are multiple winners. Therefore, 
{\footnotesize{
\[\textsc{CTR}(s,a,L_t) = (P_1+P_2)\times s, \ \textsc{CPC}(s,a,L_t) = P_1\times z_j + P_2\times a\times s, \ \textsc{SALE}(s,a,L_t) = (P_1+P_2)\times v \times s.
\]
}}
The reward of the representative bidder at $(s,a)$ pair is then:

\begin{equation}\label{eq:reward}
r(s,a,L_t) = \textsc{SALE}(s,a,L_t)-\textsc{CPC}(s,a,L_t), \quad \forall s,a,t.
\end{equation}

\section{Agnostic bidding heuristic}
\label{sec:empirical:baseline}
In this section, we present a family of basic bidding heuristic policies. These policies rely solely on auction history to propose bids likely to secure auction wins. Notably, these bidding heuristics are agnostic -- they neither estimate mean-field equilibrium bids nor consider social welfare in their recommendations. We implement these heuristics beginning with an initial population CTR distribution ($\mu_t$) and initial population bid distribution ($\alpha_t$). To maintain consistency with our repeated game model detailed in \S\ref{sec:bidding}, we assign $\mu_t=\mu, \forall t\in\mathcal{T}$.

At each time step $t$, our heuristic policy observes the population bidding distribution ($\alpha_t$) and the CTR distribution ($\mu$) and simulates $\kappa$ number of second-price auctions to estimate the winning bids. Specifically, in each auction, it samples $n$ bids from $\alpha_t$ and $n$ CTRs from $\mu$ and run a second price auction where the winner is the bidder with the highest score and pays the second highest bid adjusted by her CTR. Our agnostic heuristic bidding policy then recommends the 25th percentile and the 75th percentile of the winning bids as the lower and the upper bid ranges to bidders (referred to as bid range $[a,b]$ in Algorithm \ref{alg:baseline}). We assume that every bidder adopts the bid recommendation and bids uniformly in that range. Hence, population bidding distribution gets updated with step size $\eta$ to get closer to the target policy $\bar{\alpha}_t\sim \text{Uniform}[a, b]$. We assume that the bidders' utility ($v$) is fixed and known a priori. Algorithm \ref{alg:baseline} describes our heuristic policy.
\begin{algorithm}
\begin{algorithmic}
   \STATE {\bfseries Input:} Initialize population initial bidding distribution $\alpha_0$, population \text{CTR} distribution $\mu$, auction density $n$, utility $v$, time horizon $T$, number of intermediary auctions $\kappa>0$, and  step size $\eta \in[0,1]$
   
   \FOR{$t=1,2,\dots,T$}
   \STATE step 1: simulate $\kappa$ second price auctions to estimate winning bid distribution given $\alpha_{t-1}$ and $\mu$.
   \STATE step 2. let $a,b$ denote $25^{th}$ and $75^{th}$ percentiles of the winning bids and recommend $[a,b]$.
   \STATE step 3. let $\bar{\alpha}_t\sim \text{Uniform}[a,b]$ and update $\alpha_t=\alpha_{t-1}+\eta(\bar{\alpha}_t-\alpha_{t-1})$ with step size $\eta$.
   \ENDFOR
    \caption{Agnostic bidding heuristic (no equilibrium or social welfare)}
   \label{alg:baseline}
\end{algorithmic}
\end{algorithm}

For the following experiment, we set $\text{CTR}\in[0.01,1]$ and $\text{bids}\in[0,5]$ discretized into 20 CTR values and 20 bids (equally spaced). We also set auction density $n=30$, utility $v=5$, step size $\eta=0.7$, and the number of auctions at each time step $\kappa=10$. Figure \ref{fig:baseline} demonstrates the population CTR distribution $\mu$ set as truncated $N(0.2,0.09)$ as well as the population bidding distribution $\alpha_t$ both at $t=0$ set as truncated $N(1.5,1.44)$ and its convergent distribution at $T=1000$ averaged over 1000 independent runs. Results are consistent with different values of $\eta$. 
\begin{figure}[!ht]
     \centering
     \begin{subfigure}{0.42\textwidth}
        \includegraphics[width=\textwidth]{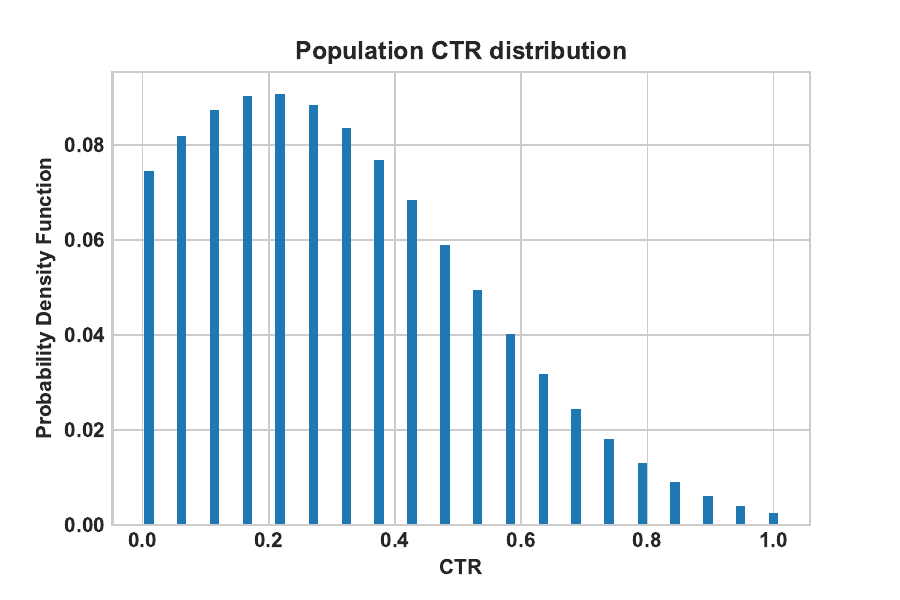}
         \caption{Population CTR Distribution}
         \label{fig:mfe_aval_ctr}
     \end{subfigure}
     \begin{subfigure}{0.42\textwidth}
        \includegraphics[width=\textwidth]{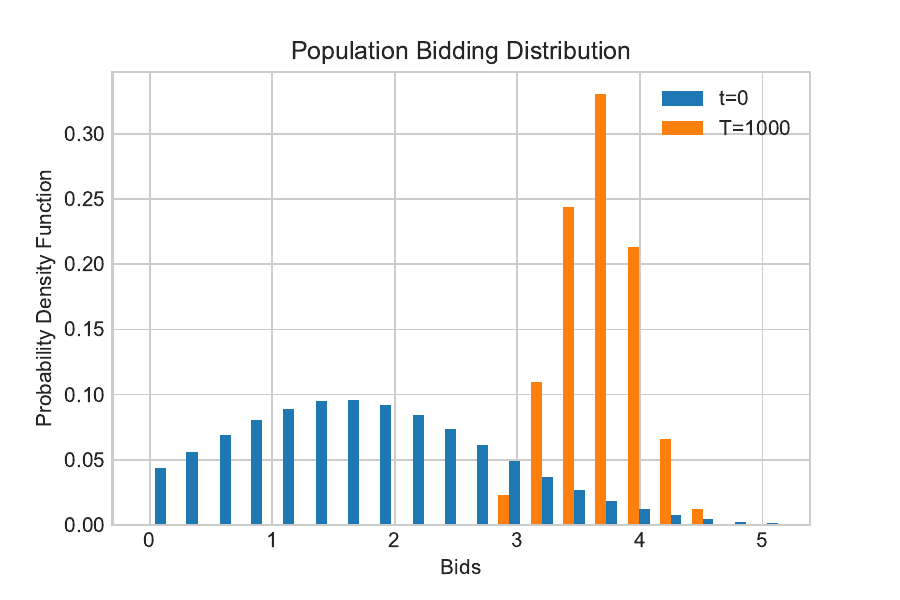}
         \caption{Population Bidding Distribution}
         \label{fig:population_baseline}
     \end{subfigure}
    \caption{Agnostic bidding heuristic. The left plot shows the population CTR distribution ($\mu$) and the right plot presents the population bidding distribution ($\alpha_t$) both at $t=0$ (set as truncated $N(1.5,1.44)$) and its convergent distribution at $T=1000$ averaged over 1000 independent runs.
    }\label{fig:baseline}
\end{figure}

\paragraph{Mean-field equilibrium dependency evaluation.}
We applied existing algorithms including Mean-Field Occupation Measure Optimization (MF-OMO) \citep{guo2022mf}, Online Mirror Descent (OMD) \citep{perolat21}, and Fictitious Play (FP) \citep{perrin2020fictitious} to estimate equilibrium bidding policies. We then evaluated all stakeholders' interests and the overall welfare and compared them against our agnostic bidding heuristic described in \ref{alg:baseline}. 

Figure \ref{fig:baseline_MFE_dists} illustrates the distribution of convergent bidding under our bidding heuristic and Mean Field Equilibrium (MFE) bidding policies, along with the population CTR distribution, which is modeled as a truncated normal distribution $N(0.2, 0.09)$. As observed, the MFE policies, factoring in the competition among advertisers, recommend a bid of \$5. This contrasts with the heuristic policy, which neglects to consider the equilibrium in the game and subsequently proposes less competitive and lower bids. The measure of exploitability, indicative of the proximity to the mean field equilibrium, is recorded at 0 for MFE policies, while it stands at 0.1356 for the heuristic policy.

\begin{figure}
    \centering
    \includegraphics[width=1.05\textwidth]{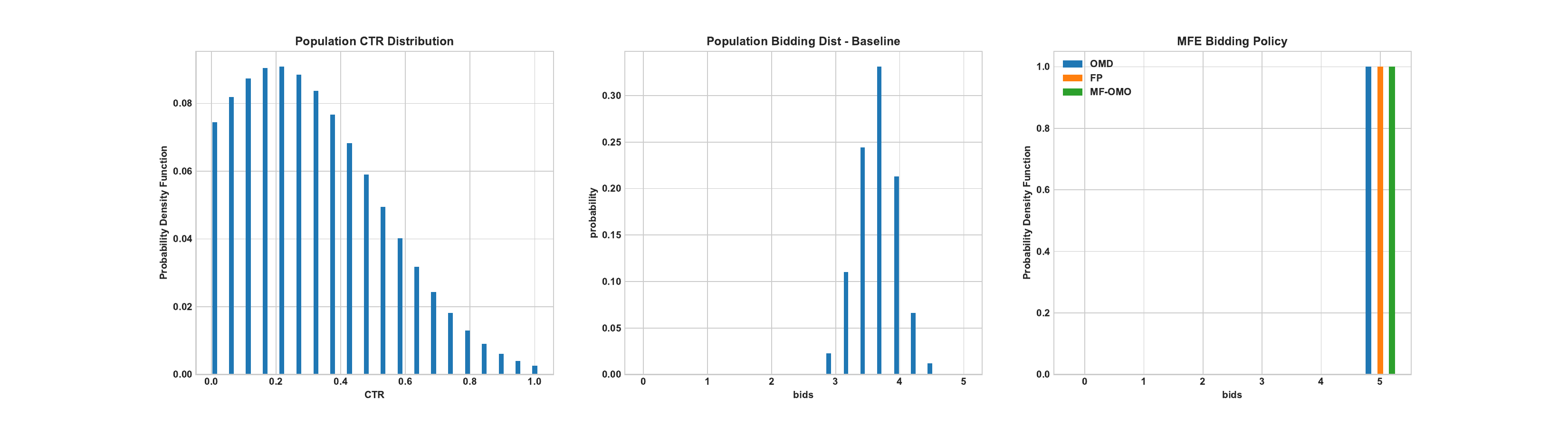}
    \caption{Population CTR distribution and bidding policies under the agnostic bidding policy and the mean field equilibrium. The left plot illustrates the distribution of the population CTR. The middle plot presents the distribution of convergent bidding under the agnostic bidding heuristic. The right plot demonstrates the mean field equilibrium bidding policy according to three mean field game algorithms (MF-OMO, OMD, and FP), each recommending a bid of \$5.}
    \label{fig:baseline_MFE_dists}
\end{figure}

Figure \ref{fig:baseline_omd} presents the welfare metrics of each involved party including the ad publisher's revenue (Cost Per Click or CPC), the advertiser's Return on Advertising Spend (RoAS), and the shopper's Click-Through Rate (CTR). The comparison is made between the heuristic and mean field equilibrium bidding policies derived from OMD, FP, and MF-OMO algorithms. In each scenario, the utility value is held constant at $v=5$. We observe that incorporating mean field equilibrium dependencies leads to more competitive bidding policies, hence a potential benefit to the ad publisher and shoppers and a potential disadvantage to the advertisers. 
\begin{figure}[!ht]
     \centering
     \begin{subfigure}{0.32\textwidth}
         \includegraphics[width=\textwidth]{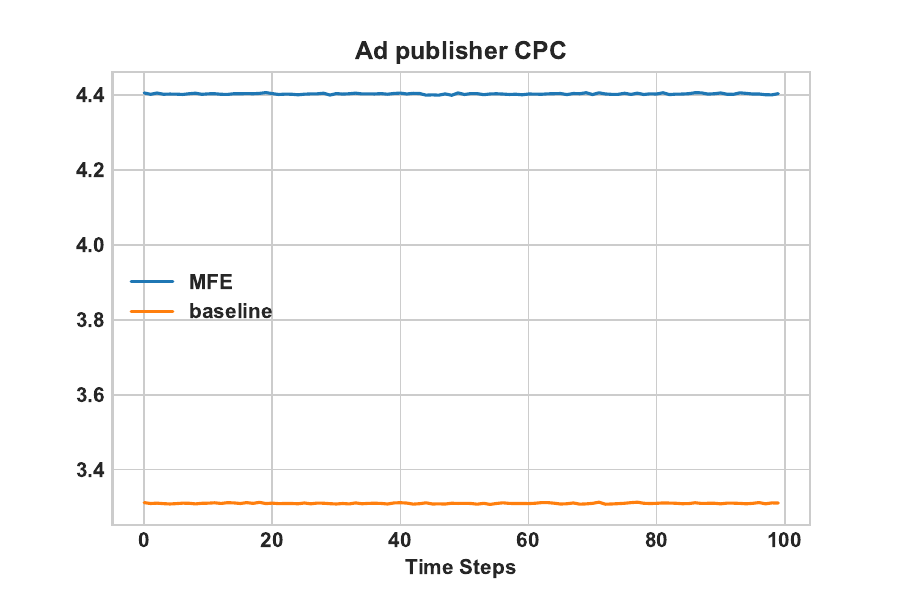}
         \caption{CPC}
         \label{fig:cpc_baseline_omd}
     \end{subfigure}
     \begin{subfigure}{0.32\textwidth}
         \includegraphics[width=\textwidth]{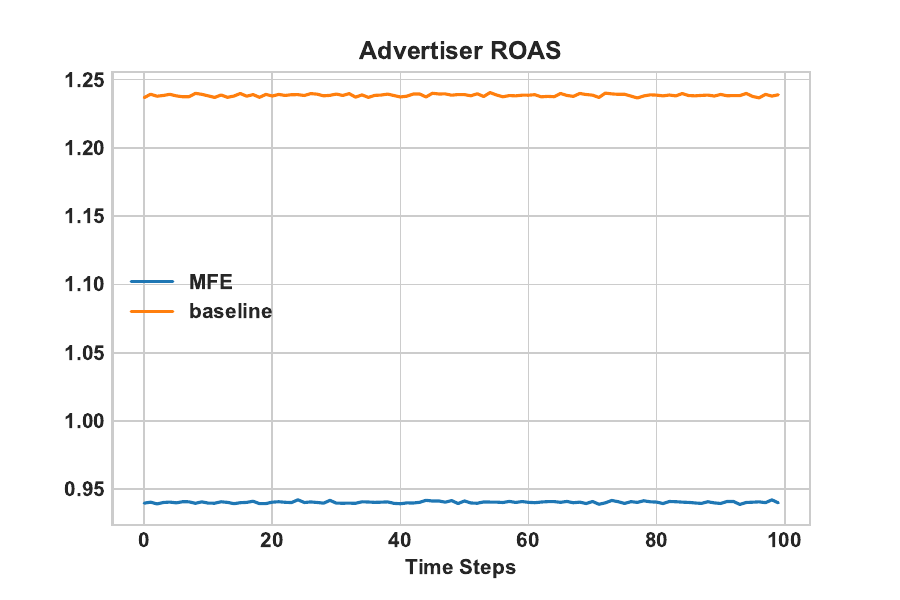}
         \caption{ROAS}
         \label{fig:RoAS_baseline_omd}
     \end{subfigure}
     \begin{subfigure}{0.32\textwidth}
         \includegraphics[width=\textwidth]{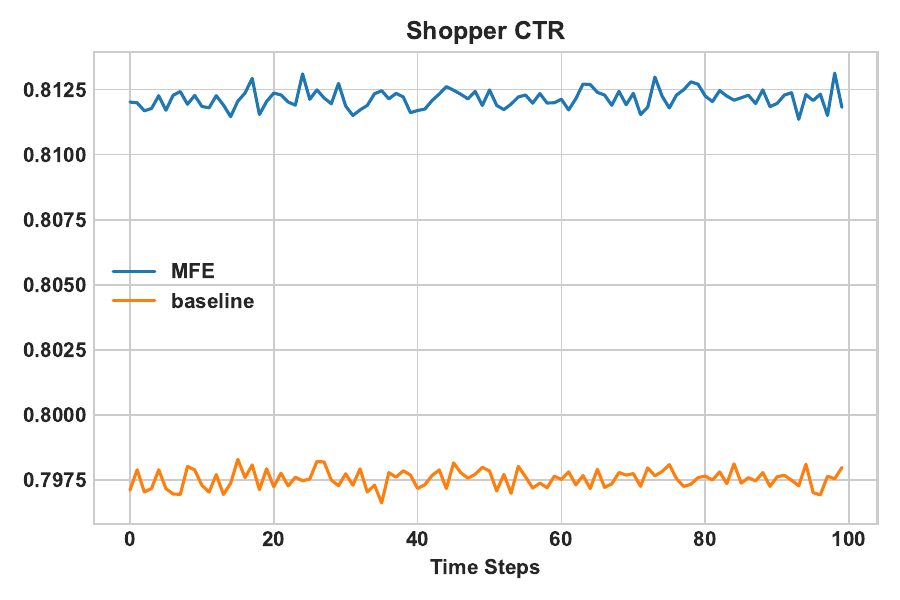}
         \caption{CTR}
         \label{fig:ctr_baseline_omd}
     \end{subfigure}
     \caption{Comparison of welfare metrics for all stakeholders under the agnostic bidding heuristic and the mean-field equilibrium bidding policy of bidding \$5: Cost Per Click (CPC) for the ad publisher, Return on Ad Spend (RoAS) for the advertisers, and the Click-Through Rate (CTR) for shoppers. Results are averaged over 100000 independent runs, each for 100 time steps.}
     \label{fig:baseline_omd}
\end{figure}

%%%%%%%%%%%%%%%%%%%%%%%%%%%%%%%%%%%%%%%%%%%%
\section{From naive scalarization to MESOB-OMO}\label{mfomo_to_mesobomo}
In this section, we provide a step-by-step derivation of  MESOB-OMO starting from the naive scalarization \eqref{scalarization}. We begin by deriving occupation measure optimization (OMO) based representations of the exploitability and the social welfare metrics, and by utilizing tools from \cite{guo2022mf}.
\subsection{OMO-based representation of exploitability}\label{review_mfomo} 

The approach of Mean-Field Occupation Measure Optimization (MF-OMO) in \citet{guo2022mf} combines  the best-response and consistency conditions in a single objective, using  the concept of occupation measures. This leads to  an efficient computing method for the NEs of MFGs. More precisely, finding NEs of an MFG with dynamics $P_t(\cdot|\cdot,\cdot,L_t)$ and rewards $r_t(\cdot,\cdot,L_t)$ ($t=0,\dots,T$) (\cf also \S\ref{sec:formulation}) is shown to be  equivalent to solving the following optimization problem:
\begin{equation}\label{mfvmo-full}
\begin{array}{ll}
\text{minimize}_{y,z,d} & \|A_dd-b\|_2^2+\|A_d^\top y+z-c_d\|_2^2+z^\top d\\
\text{subject to} & d\geq 0, \quad {\color{black}\mathbf{1}^\top \vec{d}_t=1, \,t=0,\dots,T,} \\
& \mathbf{1}^\top z\leq SA(T^2+T+2)r_{\max}, \quad z\geq 0, \tag{MF-OMO}\\
& \|y\|_2\leq \frac{S(T+1)(T+2)}{2}r_{\max}.
\end{array}
\end{equation}
The policy $\pi\in\Pi(d)$ for execution is then retrieved, with $\Pi$  defined in \S\ref{mesob-omo-main}. Then $\pi$ is an NE of the MFG if and only if the objective value of \eqref{mfvmo-full} is $0$. Here $d\in\mathbb{R}^{SA(T+1)}$ is called the occupation measure variable, which is adopted to approximate the population distribution flow with $d_{t,s,a}\approx L_t(s,a)$. As in the main text, $\vec{d}_t\in\mathbb{R}^{SA}$ is the $t$-th slicing of the tensor in its time dimension (and accordingly, $\vec{d}_t\approx L_t$). In addition, $A_d\in\mathbb{R}^{S(T+1)\times SA(T+1)}$ is a matrix-valued function of $d$ that is defined by flattened vectors of dynamics $P_t$, $c_d\in\mathbb{R}^{SA(T+1)}$ is a vector-valued function of $d$ that is defined by the flattened vectors of the rewards $r_t$, and $b\in\mathbb{R}^{S(T+1)}$ is a constant vector that depends on the initial population state distribution $\mu_0$. See the  precise formulas of these quantities in \cite{guo2022mf}. In the context of MESOB-OMO in the main text, we have $\|A_dd-b\|_2^2=g^{\text{CS}}(d)$ and $\|A_d^\top y+z-c_d\|_2^2=h^{\text{BR}}(y,z,d)$. In other words, $\|A_dd-b\|_2$ and $\|A_d^\top y+z-c_d\|_2$ correspond to consistency and best-response conditions, respectively, while the complementarity term $z^\top d$ serves as a bridge of the two conditions. 

The most notable distinction between MF-OMO and our framework is that the former concerns solely the competition, while the latter balances between competition and cooperation. Consequently, the exploitability being close to $0$ does not necessarily imply  (approximate) Pareto optimal solutions. To resolve this issue, we introduce the following lemma, which suggests that if the consistency and best-response conditions are approximately satisfied for an $y, z, L$ solution of MF-OMO, then the complementarity term ($z^\top L$) can serve as an approximate upper bound of the exploitability for the induced policy from $L$. Moreover, there exist modifications to the solutions $y$ and $z$ that approximate the consistency and best-response conditions, and yield a complementarity term which also approximate the exploitability. 

\begin{lemma}\label{zL_vs_expl} Suppose that for any $t=0,\dots,T$, $P_t(\cdot|\cdot,\cdot,L_t)$ and $r_t(\cdot,\cdot,L_t)$ are Lipschitz continuous in $L_t$.
Let $y,z,d$ be such that $d\geq 0$, $\|y\|_2 \leq y_{\max}$, $\|A_dd-b\|_2\leq \epsilon_1$, $\|A_d^\top y+z-c_d\|_2\leq \epsilon_2$, $z\geq0$ and $y\geq0$. Then for any $\pi\in\Pi(d)$, we have $z^\top d\geq \text{\rm Expl}(\pi)+O(\epsilon_1+\epsilon_2)$. Moreover, there exist $y^\pi$ and $z^\pi$ such that $\|A_d^\top y^\pi+z^\pi-c_d\|_2= O(\epsilon_1+\epsilon_2)$ and $d^\top z^\pi=\text{\rm Expl}(\pi)+O(\epsilon_1+\epsilon_2)$.
\end{lemma}

\begin{proof}
The proof can be found in Appendix \ref{appendix:lemma1}. 
\end{proof}

\subsection{OMO-based representation of social welfare} 
Here  we show how the (general utility) MFC problem of maximizing $V(\Gamma(\pi))$ over policies $\pi\in\mathcal{M}$ can be reformulated as a problem in the occupation measure variable $d$. This result is essential for the two-level game: connecting the  social planner's control problem with the Nash equilibrium of the  MFG, as the latter is an optimization problem in $d$ (and some other auxiliary variables $y,z$) but not $\pi$.  

\begin{lemma}\label{MFCinL}
For any $\epsilon\geq 0$, 
if $\pi\in\mathcal{M}$ is an $\epsilon$-suboptimal solution to $\text{maximize}_{\pi\in\mathcal{M}}\,V(\Gamma(\pi))$, then $d=\Gamma(\pi)$ is an $\epsilon$-suboptimal solution to $\text{maximize}_{d}\, V(d)$ subject to $A_dd=b,\, d\geq 0$. Conversely, if $d$ is an $\epsilon$-suboptimal solution to $\text{maximize}_{d}\, V(d)$ subject to $A_dd=b,\, d\geq 0$, then $d\in\Delta(\mathcal{S}\times\mathcal{A})$ and for any $\pi\in\Pi(d)$, $\pi$ is an $\epsilon$-suboptimal solution to $\text{maximize}_{\pi\in\mathcal{M}}\, V(\Gamma(\pi))$.
\end{lemma}
\begin{proof}
Proof is provided in Appendix \ref{appendix:lemma2}.
\end{proof}

\subsection{The MESOB-OMO reformulation of the naive scalarization}

We are now ready to derive the MESOB-OMO formulation. 
In the following, we first provide an intermediary formulation introduced in lemma \eqref{lemma:mesobomo-intermidiary}, which serves as a basis for the final penalized version presented in \eqref{mesobomo}. 

\begin{lemma}\label{lemma:mesobomo-intermidiary}
Let $y,z,d$ be an optimal solution to the following optimization problem:
\begin{equation}
\label{mesobomo-constr}
    \begin{array}{ll}
        \text{minimize}_{y,z,d} &
         -\lambda_1 V(d)+\lambda_2 z^\top d,\\
        \text{subject to} & A_dd=b, \ A_d^\top y+z=c_d,\ d\geq 0,\ z\geq 0, \\
        & {\bf 1}^\top z\leq SA(T^2+T+2)r_{\max}, \\  & \|y\|_2\leq S(T+1)(T+2)r_{\max}/2.
    \end{array}
\end{equation}
Then for any $\pi\in\Pi(d)$, $\pi$ is a Pareto-efficient solution of MESOB. Note that under the OMO notation styles, here $V(d)$ can be more explicitly written out as 
$V(d)=F(d^\top c_L^{(1)},\dots,d^\top c_L^{(K)})$, and $c_L^{(k)} = [-r_0^{(k)}(\cdot,\cdot,L_0),\dots,-r_T^{(k)}(\cdot,\cdot,L_t)]$,  where $r_t^{(k)}(\cdot,\cdot,L_t)\in\mathbb{R}^{SA}$ is seen as a flattened vector of the social metric contribution $k$. 
\end{lemma}
The key is to replace $V(\Gamma(\pi))$ in the naive scalarization \eqref{scalarization} with $V(d)+{\bf 1}_{A_d=b,d\geq 0}$ (\cf Lemma \ref{MFCinL}), while replacing $\text{Expl}(\pi)$ in \eqref{scalarization} with $z^\top L+{\bf 1}_{A_dd=b,A_d^\top y+z=c_d,(y,z,d)\in\Theta}$, where $\Theta$ denotes the constraint set of \eqref{mfvmo-full} (\cf Lemma \ref{zL_vs_expl}). The proof is rather straightforward and hence omitted. 

The constraints $A_dd=b$ and $A_d^\top y+z=c_d$ are in general non-convex. To address such issues, adding appropriate penalization terms  into the objective leads to the following MESOB-OMO formulation, which inherits the smoothness of $P_t$ and $r_t$ in the mean-field terms with essentially trivial or no constraints: %follows:
\begin{equation}\label{mesobomo_appendix}
    \begin{array}{ll}
        \text{minimize}_{y,z,d} &
        f^{\texttt{MESOB-OMO}}(y,z,d):= -\lambda_1 V(d)+\lambda_2z^\top d\\
    &\hspace{3.12cm}+\rho_1\|A_dd-b\|_2^2+\rho_2 \|A_d^\top y+z-c_d\|_2^2,\\
        \text{subject to} & d\geq 0,\quad {\bf 1}^\top\vec{d}_t=1,\quad t=0,\dots,T,\\
        &z\geq 0, \quad {\bf 1}^\top z\leq SA(T^2+T+2)r_{\max}, \\  & \|y\|_2\leq S(T+1)(T+2)r_{\max}/2.
    \end{array}
\end{equation}
Here $\rho_1,\ \rho_2>0$ correspond to the penalization parameters for the constrains $A_dd=b$ and $A_d^\top y+z=c_d$ of \eqref{mesobomo-constr} respectively. 

The following lemma, which characterizes the connection between the penalty coefficients $\rho_1$ and $\rho_2$ and the constraint violations $\|A_dd-b\|_2$ and $A_d^\top y+z-c_d\|_2$, serves as a stepping stone for the proofs of Theorems \ref{mesobomo-vs-pareto} and \ref{mesob-omo-asymptotic} below. It is a non-asymptotic extension of the standard quadratic penalty method convergence proof in the optimization literature \citep[e.g.,][]{wright1999numerical} that follows the same steps there but making use of the boundedness of all the involved variables and functions. 
\begin{lemma}\label{mesob-omo-penalty-appendix}
Suppose that the same assumptions in Theorem \ref{mesobomo-vs-pareto} on $F,\,P_t,\,r_t$ and $r_t^{(k)}$ hold. Then $\exists C_1,C_2>0$, such that for any feasible solution $(y,z,d)$ of \eqref{mesobomo} 
with 
$f^{\texttt{MESOB-OMO}}(y,z,d;\lambda_1,\lambda_2,\rho_1,\rho_2)-f^\star(\lambda_1,\lambda_2,\rho_1,\rho_2)\leq D$, 
we have $\|A_dd-b\|_2^2\leq (C_1\lambda_1+C_2\lambda_2+D)/\rho_1$ and $\|A_d^\top y+z-c_d\|_2^2\leq (\lambda_1C_1+\lambda_2C_2+D)/\rho_2$. Moreover, we also have $-\lambda_1 V_1(d)+\lambda_2 z^\top d\leq f^\star(\lambda_1,\lambda_2)+D$, where $f^\star(\lambda_1,\lambda_2)$ is the optimal objective value of the ``intermediate'' (or more constrained) version of MESOB-OMO \eqref{mesobomo-constr}.
\end{lemma}
\begin{proof}
Let $(\hat{y},\hat{z},\hat{d})$ be an arbitrary tuple such that the constraints of the \eqref{mesobomo-constr} hold. Then since the constraints of \eqref{mesobomo} is a subset of the constraints of \eqref{mesobomo-constr}, we have $f^{\texttt{MESOB-OMO}}(y,z,d;\lambda_1,\lambda_2,\rho_1,\rho_2)\geq f^\star(\lambda_1,\lambda_2,\rho_1,\rho_2)$, and hence
\begin{equation}\label{mesob-omo-penalty-proof-key-ineq}
\begin{split}
-&\lambda_1V(d)+\lambda_2z^\top d+\rho_1 \|A_dd-b\|_2^2+\rho_2\|A_d^\top y+z-c_d\|_2^2\\
&\leq -\lambda_1V(\hat{d}) + \lambda_2 \hat{z}^\top\hat{d}+\rho_1\|A_{\hat{d}}\hat{d}-b\|_2^2+\rho_2\|A_{\hat{d}}^\top \hat{y}+\hat{z}-c_{\hat{d}}\|_2^2+D\\
&=-\lambda_1 V_1(\hat{d}) +\lambda_2 \hat{z}^\top\hat{d}+ D,
\end{split}
\end{equation}
which implies that (due to the non-negativity of norms)
\[
-\lambda_1 V(d)+\lambda_2 z^\top d\leq -\lambda_1 V_1(\hat{d}) +\lambda_2 \hat{z}^\top\hat{d}+ D.
\]    
Hence we have $-\lambda_1 V(d)+\lambda_2 z^\top d\leq f^\star(\lambda_1,\lambda_2)+D$ by the fact that $(\hat{y},\hat{z},\hat{d})$ is an arbitrary tuple for which the constraints of \eqref{mesobomo-constr} hold. 

In addition, again by \eqref{mesob-omo-penalty-proof-key-ineq}, we also have 
\[
\rho_1 \|A_dd-b\|_2^2+\rho_2\|A_d^\top y+z-c_d\|_2^2\leq 2\lambda_1 \sup_{\{\vec{d}_{t}\}_{t=0}^T\subseteq\Delta(\mathcal{S}\times\mathcal{A})}V(d)+2\lambda_2 \sup_{z\in Z,\{\vec{d}_{t}\}_{t=0}^T\subseteq\Delta(\mathcal{S}\times\mathcal{A})}z^\top d + D.
\]
Now since $C_1:=2\sup_{\{\vec{d}_{t}\}_{t=0}^T\subseteq\Delta(\mathcal{S}\times\mathcal{A})}V(d)<\infty$ and $C_2:=2\sup_{z\in Z,\{\vec{d}_{t}\}_{t=0}^T\subseteq\Delta(\mathcal{S}\times\mathcal{A})}z^\top d<\infty,$ by the continuity of $F,\,P,\,r_t$ and $r_t^{(k)}$ and the compactness of $\Delta(\mathcal{S}\times\mathcal{A})$ and $Z$, we arrive at the desired claims. %the proof is complete. 
\end{proof}

%%%%%%%%%%%%%%%%%%%%%%%%%%%%%%%%%%%%%%%%%%%%
\section{Technical proofs}
\label{proofs}
In this section, we provide technical proofs of Lemmas \ref{zL_vs_expl} and \ref{MFCinL}, along with Theorems \ref{mesobomo-vs-pareto} and \ref{mesob-omo-asymptotic}. In the interest of consistency with the notation used in \cite{guo2022mf}, we employ $L$ as the occupation measure variable, instead of $d$.

\subsection{Proof of Lemma \ref{zL_vs_expl}}
\label{appendix:lemma1}

\begin{proof}
Let $\Delta_1:=A_LL-b$ and $\Delta_2:=A_L^\top y+z-c_L$. Then 
\[
z^\top L=(c_L-A_L^\top y+\Delta_2)^\top L=c_L^\top L-y^\top A_LL+\Delta_2^\top L=c_L^\top L-b^\top y+\Delta_2^\top L-\Delta_1^\top y.
\]
Now for any $\pi\in\Pi(L)$, define $L^\pi=\Gamma(\pi)$. Then by \cite[Proof of Theorem 8, Step 3]{guo2022mf},
\begin{equation}\label{L_pi_minus_L}
\sum_{s\in\mathcal{S},a\in\mathcal{A}}|L_t^\pi(s,a)-L_t(s,a)|\leq \dfrac{(C_P+1)^{t+1}-1}{C_P}\sqrt{S}\epsilon_1,
\end{equation}
where $C_P$ is the Lipschitz constant of the transitions. Since $c_{L^\pi}^\top L^\pi=-V_{\mu_0}^\pi$ by definition of $c_L$ and $L^\pi$, we get
\begin{align*}
    |c_L^\top L+V_{\mu_0}^\pi(L^\pi)|&=|c_L^\top L-c_{L^\pi}^\top L^\pi|\leq |c_L^\top L-c_L^\top L^{\pi}|+ |c_L^\top L^\pi -c_{L^\pi }^\top L^{\pi}|\\
    & \leq (r_{\max}+C_r) \sum_{t=0}^T \|L_t - L^\pi_t \|_1 \\
    & \leq (r_{\max} + C_r)\sum_{t=0}^T\dfrac{(C_P+1)^{t+1}-1}{C_P}\sqrt{S}\epsilon_1    =O(\epsilon_1),
\end{align*}
where $r_{\max} = \sup_{s,a,L} r(s,a,L)$. Now by the definition of $b$, $b^\top y=\mu_0^\top y_T$, where $y=[y_0,\dots,y_T]$. Consider $y^\pi, z^\pi$ defined as the $\hat{y}$ and $\hat{z}$ corresponding to $L^\pi$ (as in \cite[Proposition 6]{guo2022mf}). Then $y^\pi_T=-V_0^\star(L^\pi)$. Since $L^\pi$ and $L$ differ by $O(\epsilon_1)$ (see \eqref{L_pi_minus_L}), by the Lipschitz continuity of $r$ and $P$, Cauchy-Schwarz inequality, and the property $\|.\|_2 \leq \|.\|_1$, we see
\begin{align*}
 \|A_{L^\pi}^\top y+z-c_{L^\pi}\|_2 &\leq \|A_{L}^\top y+z-c_{L}\|_2 + \|y\|_2\|A_{L^\pi} - A_L\|_2 + \|c_{L^\pi} - c_L\|_2 \\
 & \leq \epsilon_2 + (\|y\|_2\ C_p+C_r) \sum_{t=0}^T\|L_t^\pi - L_t\|_1 \\
 & \leq  \epsilon_2 + (y_{\max} C_p+C_r) \sum_{t=0}^T\dfrac{(C_P+1)^{t+1}-1}{C_P}\sqrt{S}\epsilon_1 = O(\epsilon_1+\epsilon_2).
\end{align*}

Hence by adapting \cite[Proof of Proposition 6]{guo2022mf} and adding $O(\epsilon_1+\epsilon_2)$ to the RHS of all the inequalities that were used to show $\hat y_T \geq y_T$, one can prove that $y^{\pi}_T\geq  y_T+O(\epsilon_1+\epsilon_2)$, and hence 
\[
-V_{\mu_0}^\star(L^\pi)=\mu_0^\top y^{\pi}_T\geq \mu_0^\top y_T+O(\epsilon_1+\epsilon_2)=b^\top y+O(\epsilon_1+\epsilon_2).
\]

Therefore we obtain the bound on $z^\top L$ as, 
\[
z^\top L = c_L^\top L-b^\top y + O(\epsilon_1+\epsilon_2)\geq -V_{\mu_0}^{\pi}(L^\pi)+V_{\mu_0}^\star(L^\pi)+O(\epsilon_1+\epsilon_2)=\text{Expl}(\pi)+O(\epsilon_1+\epsilon_2).
\]

Finally, since $A_{L^\pi}y^\pi+z^{\pi}-c_{L^\pi}=0$ by definition, we have 
$A_Ly^\pi+z^\pi-c_L=O(\epsilon_1+\epsilon_2)$ 
due to \eqref{L_pi_minus_L}. Also, $L^\top z^\pi = (L^\pi )^\top z^\pi + L^\top z^\pi  - (L^\pi )^\top z^\pi = \text{Expl}(\pi) + (L^\pi - L )^\top z^\pi $, where the second equality follows from \cite[Proof of Theorem 9]{guo2022mf}. By definition $z^\pi$ is bounded. Thus, following \eqref{L_pi_minus_L},  $L^\top z^\pi = \text{Expl}(\pi) + O(\epsilon_1+\epsilon_2)$.
\end{proof}

\subsection{Proof of Lemma \ref{MFCinL}}
\label{appendix:lemma2}

\begin{proof}
Let $\pi^\star$ be an optimal policy that maximizes $V(\Gamma(\pi^\star))$ over $\pi\in\mathcal{M}$, and $L^\star$ be an optimal population flow that maximizes $V(L)$ subject to $A_LL=b,\,L\geq 0$. Note that $\pi^\star$ exists as $V$ and $\Gamma$ are both continuous, and $\mathcal{M}$ is a compact set in $\mathbb{R}^{S\times A\times (T+1)}$. Similarly, $L^\star$ exists as $V$ is continuous and $A_LL=b$ implies that ${\bf 1}^\top L=1$. Hence, with $L\geq0$ and the continuity of $P_t$ in $L_t$,  the constraint set is closed and bounded (and hence compact). 

Now for the first statement, since $L=\Gamma(\pi)$, we have $A_LL=b,\,L\geq 0$ by the definition of $\Gamma(\cdot)$ and the construction of $A_L$ and $b$. Hence  $V(L^\star)\geq V(L)=V(\Gamma(\pi))\geq V(\Gamma(\pi^\star))-\epsilon$. Meanwhile,  for any $\pi_{L^\star}\in\Pi(L^\star)$, we have $\pi_{L^\star}\in\mathcal{M}$ and $L^\star=\Gamma(\pi_{L^\star})$ by the definitions of $\Pi(\cdot)$ and $\Gamma(\cdot)$, and hence $V(L^\star)=V(\Gamma(\pi_{L^\star}))\leq V(\Gamma(\pi^\star))$. Putting these inequalities together, we see that $V(L^\star) -V(L)\leq V(\Gamma(\pi^\star))-(V(\Gamma(\pi^\star))-\epsilon)=\epsilon$, and hence $L$ is $\epsilon$-suboptimal for maximizing $V(L)$ subject to $A_LL=b,\,L\geq0$. 

Similarly, for the second statement, note that $A_LL=b$ and $L\geq 0$, implying $L\in(\Delta(\mathcal{S}\times\mathcal{A}))^T$ and hence $\Pi(L)\subseteq\mathcal{M}$. Now for any $\pi\in\Pi(L)$, by the definition of $\Pi(L)$ and the fact that $A_LL=b$, we have $\pi\in\mathcal{M}$ and $L=\Gamma(\pi)$. Let $L^{\pi^\star}=\Gamma(\pi^\star)$. Then $A_{L^{\pi^\star}}L^{\pi^\star}=b,\, L^{\pi^\star}\geq 0$, therefore $V(\Gamma(\pi^\star))=V(L^{\pi^\star})\leq V(L)\leq V(L^\star)$ and  $V(\Gamma(\pi^\star))\geq V(\Gamma(\pi))=V(L)\geq V(L^\star)-\epsilon$, which imply that $V(\Gamma(\pi^\star))-V(\Gamma(\pi))\leq V(L^\star) - (V(L^\star)-\epsilon)=\epsilon$. Hence $\pi$ is $\epsilon$-suboptimal for maximizing $V(\Gamma(\pi))$ over $\mathcal{M}$.  
\end{proof}

\subsection{Proof of Theorem \ref{mesobomo-vs-pareto}}
\label{appendix:theorem4}

\begin{proof}
To better illustrate the main idea of the proof, we prove the claims for the limiting case when $\rho$ goes to infinity, in which case we would have $||A_LL-b||=0$ and $||A_L^\top y+z-c_L||=0$ as indicated by Lemma \ref{mesob-omo-penalty-appendix}. More precisely, we first consider the limiting case when we solve the ``intermediate'' MESOB-OMO \eqref{mesobomo-constr} to $\epsilon$-suboptimality. Accordingly, we denote $f^{\text{MESOB-OMO-constr}}(y,z,L)$ and $f^\star(\lambda_1,\lambda_2)$ as the objective and the optimal objective value of \eqref{mesobomo-constr}, respectively. The case when $\rho$ is finite but satisfies the lower bound specified in the theorem can be similarly derived by combining Lemma \ref{zL_vs_expl} and Lemma \ref{mesob-omo-penalty-appendix}. 

Suppose on the contrary that there exists a policy $\pi'$ such that 
\begin{equation}\label{contradict-assump}
V(\Gamma(\pi')) \geq V(\Gamma(\pi)) + \epsilon/(2\lambda_1) \quad \text{and} \quad \text{Expl}(\pi') \leq \text{Expl}(\pi) - \epsilon/(2\lambda_2),
\end{equation}
where one of the inequalities is strict. Let $L'=\Gamma(\pi')$ and define $y'$ and $z'$ as in \citet[Theorem 9]{guo2022mf}. Then $y', z', L'$ is a feasible solution of the ``intermediate'' MESOB-OMO given in (\ref{mesobomo-constr}) and $L'^\top z'=\text{Expl}(\pi')$ by \citet[Theorem 9]{guo2022mf}.

Note that since we have $||A_LL-b||=0$ and $||A_L^\top y+z-c_L||=0$, by applying Lemma \ref{zL_vs_expl}, we see that $L^\top z\geq \text{Expl}(\pi)$. Hence we have that 
\[
\begin{split}
f^\texttt{MESOB-OMO-constr}(y,z,L) - f^\star(\lambda_1,\lambda_2)&\geq f^\texttt{MESOB-OMO-constr}(y,z,L) - f^\texttt{MESOB-OMO-constr}(y',z',L') \\
&= -\lambda_1 V(L)+\lambda_2\ L^\top z  + \lambda_1 V(L')+\lambda_2\ L'^\top z' \\
&= \lambda_1\ [V(L')-V(L)]+\lambda_2\  [L^\top z-L'^\top z'] \\
&\geq \lambda_1\ [V(L')-V(L)]+\lambda_2\  [\text{Expl}(\pi)-\text{Expl}(\pi')]\\
&> \lambda_1\ (\frac{\epsilon}{2\lambda_1}) + \lambda_2\ (\frac{\epsilon}{2\lambda_2}) = \epsilon,
\end{split}
\]
where the strict inequality follows from our assumptions in \eqref{contradict-assump}. But this obviously contradicts with the assumption that $(y,z,L)$ solves \eqref{mesobomo-constr} to $\epsilon$-suboptimality, thus proving our claim. 
\end{proof}

\subsection{Proof of Theorem \ref{mesob-omo-asymptotic}}
We need the following lemma, which bounds $V(d)$ when $d$ is approximately consistent. 
\begin{lemma}\label{mfc_approx_consistent_ub}
If $\|A_dd-b\|_2\leq \epsilon$ and $d\geq 0$, then $V(d)\leq \max_{\pi\in\mathcal{M}}V(\Gamma(\pi))+O(\epsilon)$.
\end{lemma}
\begin{proof}
This can be proved by defining $d^\pi:=\Gamma(\pi)$ for some arbitrary $\pi\in\Pi(d)$, and invoking \cite[Proof of Theorem 8, Step 3]{guo2022mf} as in the proof of Lemma \ref{zL_vs_expl}. The details are hence omitted for brevity. 
\end{proof}
We are now ready to prove Theorem \ref{mesob-omo-asymptotic}.
\begin{proof}
Without loss of generality, let's assume that 
$\epsilon^l\leq 1$   %$\lambda_1^l,\,\lambda_2^l,\,\epsilon^l\leq 1$ 
for all $l\geq 0$. Then 
by Lemma \ref{mesob-omo-penalty-appendix}, we have 
\[
\|A_{d^l}d^l-b\|_2^2\leq \dfrac{(\lambda_1^l C_1+\lambda_2^lC_2+\epsilon^l)\epsilon^l}{2\max\{\lambda_1^l,\lambda_2^l,1\}}\leq \dfrac{(C_1+C_2+1)\epsilon^l}{2},
\]
\begin{equation}\label{consistency_lambda1_bound}
\|A_{d^l}d^l-b\|_2^2\leq \dfrac{(\lambda_1^l C_1+\lambda_2^lC_2+\epsilon^l)\epsilon^l}{2\max\{(\lambda_1^l)^2,\lambda_1^l\lambda_2^l,\lambda_1^l\}}\leq \dfrac{(C_1+C_2+1)\epsilon^l}{2\lambda_1^l},
\end{equation}
and similarly 
\[
\|A_{d^l}^\top y^l+z^l-c_{d^l}\|_2^2\leq \dfrac{(C_1+C_2+1)\epsilon^l}{2}.
\]
In addition, we have $-\lambda_1^l V_1(d^l)+\lambda_2^l (z^l)^\top d^l\leq f^\star(\lambda_1^l,\lambda_2^l)+\epsilon^l$. 

\paragraph{Proof of the first claim.}
Now by Lemma \ref{MFCinL} and \eqref{mfvmo-full}, we have that $\pi\in\mathcal{M}$ is a solution to the equilibrium selection problem that maximizes $V(\Gamma(\pi))$ over all NE policies $\pi$ with $\text{Expl}(\pi)=0$ if and only if $\exists\,y,z,d$, such that $\pi\in\Pi(d)$ and $(y,z,d)$ solves the following optimization problem:
\begin{equation}
\label{mesobomo-equilibrium-selection}
    \begin{array}{ll}
        \text{minimize}_{y,z,d} &
         -V(d),\\
        \text{subject to} & z^\top d=0,\, %\text{Expl}(\pi)=0,\quad \pi\in\Pi(d),\\
        A_dd=b, \ A_d^\top y+z=c_d,\ d\geq 0,\ z\geq 0, \tag{equilibrium selection}\\
        & {\bf 1}^\top z\leq SA(T^2+T+2)r_{\max}, \\  & \|y\|_2\leq S(T+1)(T+2)r_{\max}/2.
    \end{array}
\end{equation}
Let $(\tilde{y},\tilde{z},\tilde{d})$ be an optimal solution of \eqref{mesobomo-equilibrium-selection}. 
 Hence by the fact that $\inf_{l\geq 0}\lambda_1^l>0$, we have 
\[
-V(d^l)+\frac{\lambda_2^l}{\lambda_1^l}(z^l)^\top d^l\leq -V(d^\star)+\frac{\lambda_2^l}{\lambda_1^l}(z^\star)^\top d^\star+O(\epsilon^l)\leq -V(\tilde{d})+O(\epsilon^l),
\]
where $(y^\star,z^\star,d^\star)$ solves \eqref{mesobomo-constr} for some $y^\star$, namely $-\lambda_1^l V(d^\star)+\lambda_2^l (z^\star)^\top d^\star=f^\star(\lambda_1^l,\lambda_2^l)$. Here the second inequality comes from the fact that $(\tilde{y},\tilde{z},\tilde{d})$ is also feasible for \eqref{mesobomo-constr}. Hence by the non-negativity of $z^l$ and $d^l$, we have 
\[
-V(d^l)\leq -V(\tilde{d})+O(\epsilon^l) \quad \text{and} \quad (z^l)^\top d^l\leq \frac{\lambda_1^l}{\lambda_2^l}O(1).
\]
Hence when $\lim\limits_{l\rightarrow\infty}\frac{\lambda_1^l}{\lambda_2^l}=0$ and $\lim\limits_{l\rightarrow\infty}\epsilon^l=0$, we have for any limit point $(\bar{y},\bar{z},\bar{d})$ of $(y^l,z^l,d^l)$, due to the continuity assumptions on $F,\,P_t,\,r_t,\,r_t^{(k)}$, that $A_{\bar{d}}\bar{d}=b$, $A_{\bar{d}}^\top \bar{y}+\bar{z}=c_{\bar{d}}$, $\bar{z}^\top \bar{d}=0$, $V(\bar{d})\geq V(\tilde{d})$, and also $\bar{d}\geq 0,\,\bar{z}\geq 0$, ${\bf 1}^\top \bar{z}\leq SA(T^2+T+2)r_{\max}$ and $\|\bar{y}\|_2\leq S(T+1)(T+2)r_{\max}/2$. Together with the optimality of $\tilde{d}$ for \eqref{mesobomo-equilibrium-selection}, we see that $\bar{d}$ also solves \eqref{mesobomo-equilibrium-selection}, from which we conclude that for any $\bar{\pi}\in\Pi(\bar{d})$, it solves the equilibrium selection problem. 

\paragraph{Proof of the second claim.}
Similarly, define $V^{\max}:=\max_{\pi\in\mathcal{M}}V(\Gamma(\pi))$. Then  by Lemma \ref{zL_vs_expl}, Lemma \ref{MFCinL} and \eqref{mfvmo-full}, we have that $\pi$ solves  the  social equalizing problem that minimizes $\text{Expl}(\pi)$ over all socially optimal policies that maximizes $V(\Gamma(\pi))$ if and only if $\exists\,y,z,d$, such that $\pi\in\Pi(d)$ and $(y,z,d)$ solves the following optimization problem:
\begin{equation}
\label{mesobomo-social-equalizing}
    \begin{array}{ll}
        \text{minimize}_{y,z,d} &
         z^\top d,\\
        \text{subject to} & V(d)=V^{\max},\,A_dd=b, \ A_d^\top y+z=c_d,\ d\geq 0,\ z\geq 0, \tag{social equalizing}\\
        & {\bf 1}^\top z\leq SA(T^2+T+2)r_{\max}, \\  & \|y\|_2\leq S(T+1)(T+2)r_{\max}/2.
    \end{array}
\end{equation}
With some slight abuse of notation, we again denote 
$(\tilde{y},\tilde{z},\tilde{d})$ as the optimal solution to \eqref{mesobomo-social-equalizing}. Then by the fact that $\inf_{l\geq 0}\lambda_2^l>0$, we have
\begin{equation}\label{social_equalizing_penalty_ineqs}
(z^l)^\top d^l-\frac{\lambda_1^l}{\lambda_2^l}V(d^l)\leq (z^\star)^\top d^\star-\frac{\lambda_1^l}{\lambda_2^l}V(d^\star)+O(\epsilon^l)\leq \tilde{z}^\top \tilde{d}-\frac{\lambda_1^l}{\lambda_2^l}V^{\max}+O(\epsilon^l),
\end{equation}
where $(y^\star,z^\star,d^\star)$ solves \eqref{mesobomo-constr} for some $y^\star$ as in the proof of the first claim above. Here the second inequality comes from the fact that $V(\tilde{d})=V^{\max}$ and that $(\tilde{y},\tilde{z},\tilde{d})$ is also feasible for \eqref{mesobomo-constr}. By the bounds on $z^l,\,d^l,\,\tilde{z},\,\tilde{d}$, we immediately see that 
\begin{equation}\label{vdiff_ineq}
V^{\max}-V(d^l)\leq \frac{\lambda_2^l}{\lambda_1^l}\ O(1).
\end{equation}
Hence when $\lim\limits_{l\rightarrow\infty}\frac{\lambda_1^l}{\lambda_2^l}=\infty$ and $\lim\limits_{l\rightarrow\infty}\epsilon^l=0$,  for any limit point $(\bar{y},\bar{z},\bar{d})$ of $(y^l,z^l,d^l)$, again due to the continuity assumptions of $F,\,P_t,\,r_t,\,r_t^{(k)}$, we have $A_{\bar{d}}\bar{d}=b$, $A_{\bar{d}}^\top \bar{y}+\bar{z}=c_{\bar{d}}$, and also $\bar{d}\geq 0,\,\bar{z}\geq 0$, ${\bf 1}^\top \bar{z}\leq SA(T^2+T+2)r_{\max}$ and $\|\bar{y}\|_2\leq S(T+1)(T+2)r_{\max}/2$. Particularly, by Lemma \ref{MFCinL}, $A_{\bar{d}}\bar{d}=b,\,\bar{d}\geq 0$ implies that $V(\bar{d})\leq V^{\max}$. Since we also have $V^{\max}-V(\bar{d})\leq 0$ by taking the limit of \eqref{vdiff_ineq}, we conclude that $V(\bar{d})=V^{\max}$. 
 
Finally, by Lemma \ref{mfc_approx_consistent_ub} and \eqref{consistency_lambda1_bound}, we have 
\[
V(d^l)\leq V^{\max}+O(\epsilon^l/\lambda_1^l), 
\]
and hence we have by \eqref{social_equalizing_penalty_ineqs} that
\[
(z^l)^\top d^l\leq \tilde{z}^\top \tilde{d}+O(\epsilon^l)+\dfrac{\lambda_1^l}{\lambda_2^l}(V(d^l)-V^{\max})\leq \tilde{z}^\top \tilde{d}+O(\epsilon^l)+\dfrac{1}{\inf_{l\geq 0}\lambda_2^l}O(\epsilon^l)=\tilde{z}^\top\tilde{d}+O(\epsilon^l).
\]
Now taking the subsequence limit on both sides of the above inequality, we have $\bar{z}^\top \bar{d}\leq \tilde{z}^\top\tilde{d}$. Together with the optimality of $\tilde{d}$ and $\tilde{z}$ for \eqref{mesobomo-social-equalizing}, we see that $(\bar{y},\bar{z},\bar{d})$ also solves \eqref{mesobomo-social-equalizing}, from which we conclude that for any $\bar{\pi}\in\Pi(\bar{d})$, it solves the social equalizing problem.
\end{proof}

\end{document}